\documentclass[11pt]{article}
\usepackage{amsthm}
\usepackage{graphicx,epsf,subfigure}
\usepackage{amsmath,amsfonts, natbib,verbatim,booktabs,longtable,multirow,lscape}
\usepackage{enumerate, color, url}
\usepackage{dsfont,ctable}
\usepackage{booktabs}
\usepackage[margin = 0.75in]{geometry}
\usepackage{setspace}
\usepackage{algorithm}

\newtheorem{proposition}{Proposition}

\newcommand{\bt}{\mathbf{t}}

\newcommand{\bR}{\mathbf{R}}

\newcommand{\cG}{\mathcal{G}}

\newcommand{\bmu}{\boldsymbol{\mu}}

\newcommand{\bone}{\boldsymbol{1}}

\newcommand{\argmin}{\operatornamewithlimits{argmin}}

\begin{document}

\title{A backward procedure for change-point detection with applications to copy number variation detection}
\author{Seung Jun Shin, Yichao Wu, and Ning Hao \smallskip \\ 
\it Korea University, University of Illinois at Chicago, and The University of Arizona}
\date{}

\maketitle

\abstract{
Change-point detection regains much attention recently for analyzing array or sequencing data for copy number variation (CNV) detection. In such applications, the true signals are typically very short and buried in the long data sequence, which makes it challenging to identify the variations efficiently and accurately. In this article, we propose a new change-point detection method, a backward procedure, which is not only fast and simple enough to exploit high-dimensional data but also performs very well for detecting short signals. Although motivated by CNV detection, the backward procedure is generally applicable to assorted change-point problems that arise in a variety of scientific applications. It is illustrated by both simulated and real CNV data that the backward detection has clear advantages over other competing methods especially when the true signal is short. \smallskip \\ 
\textsc{Keywords}: 
Backward detection; 
copy number variation; 
mean change-point model; 
multiple change points;
Short signal.}

\section{Introduction}
As a classic topic, change-point detection regains much attention recently in the context of uncovering structural change in big data. In particular, a normal mean change-point model and its variants have been applied to analyze high-throughput data for DNA copy number variation (CNV) detection. The CNV is defined as duplication or deletion of a segment of DNA sequences compared to the reference genome, and can cause significant effects at molecular levels and be associated with susceptibility (or resistance) to disease \citep{feuk2006structural,freeman2006copy,mccarroll2007copy}.

There are multiple sources of data that provide us with copy number information.  The microarray Comparative Genomic Hybridization (aCGH) techniques have been widely used for CNV detection \citep{urban2006high}. The aCGH is helpful to detect long CNV segments with tens of kilobases (kb) or more, but is not able to locate small-scale CNVs with length shorter than its minimal resolution ($>$10 kb), which are common in the human genome \citep{sebat2004large,carter2007methods,wong2007comprehensive}. In addition to the aCGH approach, the single nucleotide polymorphism (SNP) genotypeing array has become an alternative in CNV detection due to its improved resolution \citep{citeulike:5084099}. For example, popular SNP array platforms such as Illumina \citep{Peifferetal}  and Affymetrix \citep{McCarrolletal} allow to detect CNVs with kilobase-resolution.
In SNP arrays, CNV information is measured by the total fluorescent intensity signal ratios from both alleles at each SNP locus referred to as log-R-ratio (LRR). It also allows to have the relative ratio of the fluorescent signals between two alleles, known as B allele frequency (BAF). Finally, in very recent applications, aligned DNA sequencing data with even higher resolution can be directly used for CNV detection. The next generation sequencing (NGS) techniques typically  produce a millions of short reads that are to be aligned with the reference genome. Both of the associated read depth (RD) and distances of paired-end (DPE) from the aligned sequence are important sources of inferring CNV \citep{medvedev2009computational,abyzov2011cnvnator,duan2013comparative,chen2017allele}. Note that there is a trade-off between the resolution and data size. With higher resolution data, it is possible to discover shorter CNVs; at the same time, the larger data size  brings great challenge in computation. This might be one of the reasons why the SNP array has been most popular in recent CNV studies since aCGH data have low resolutions and RD or DPE data from NGS are too large to be handled directly. However, as related computing technologies advance, the NGS data are getting more attentions in the recent applications. Finally, we summarize popular data sources for CNV detection in Table \ref{tb::data}.

\begin{table}[h]
\caption{Data types for CNV information and the corresponding resolutions.} \label{tb::data}\begin{center}
\begin{tabular} {llc} \hline
\multicolumn{1}{c}{Data} & \multicolumn{1}{c}{Measurement} & \multicolumn{1}{c}{Resolution} \\ \hline
array CGH & Fluorescence intensity ratio & $10 \sim 100 \times$kb \\
SNP array & Fluorescence intensity ratio/B-allele frequency        & kb \\
NGS  & Read depth/Distance of paired end & b \\ \hline
\end{tabular}
\end{center}
\end{table}

There have been tons of methods developed for CNV detection. Different approaches are applied to different types of data sources. As one of the most popular approaches, the CNV detection problem can be regarded as an application of the change-point model which has been actively studied in statistics. For example, both of the LRR from SNP array and $\log_2$ ratio from aCGH has a mean value of zero for normal copy number, while negative (resp. positive) for the deletion (resp. duplication). Similar ideas can also be employed for the RD data in the sense that one may observe less (resp. more) read counts in a region with deleted (resp. duplicated) copy number. In all the examples, the data structure changes at CNVs. Naturally, one may infer CNVs by checking the subregions where the LRR or read-depth is significantly different from the mean of rest.

The change-point model has a long history that traces back to the 1950s. See \cite{Page:1955,Page:1957,ChernoffZacks:1964,Gardner:1969} and \cite{SenSri:1975} for the early developments of change-point model with at most one change point. The data size considered in those papers is also small. However, new applications call for more flexible models capable of detecting multiple change points scattered along a huge sequence. Recent developments include circular binary segmentation \citep[CBS,][]{OVLW:2004,VO:2007}, $\ell_1$ penalization \citep{HuangEtal:2005,fusedlasso:2008,ZhangLange:2010} and total-variation-penalized estimation \citep[TVP,][]{harchaoui2010multiple}, fragment assembling algorithm \citep[FASeg,][]{yu2007forward}, screening and ranking algorithm \citep[SaRa,][]{NiuZhang:2010,HaoNiuZhang:2013}, likelihood ratio selection \citep[LRS,][]{CaiJengLi:2010}, simultaneous multiscale change point estimator \citep[SMUCE,][]{frick2014multiscale}, and wild binary segmentation \citep[WBS,][]{fryzlewicz2013wild} among  many others. Hidden Markov model is another popular approach for CNV detection \citep{Fridlyand:HMM:2004,Wangetal:07,szatkiewicz2013improving}. Yet it relies on some application specific assumptions valid only for certain copy number data. \cite{Zhang:2010} provided a comprehensive overview on CNV detection as an application of the change-point model. \cite{roy2013evaluation} compared the performance of a few recent CNV detection methods under various scenarios.

As an application of the change-point model, inferring CNV is regarded as a very challenging problem since the CNVs subregions are usually very short and hidden in a very long sequence. The size of detectable CNVs typically range from 1,000 bps to megabases. \cite{international2010integrating} show that the average size of total CNVs in the individual genome is $3.5 \pm 0.5$ Mbp (0.1\%). As an illustration, we analyze the SNP array data collected from the Autism Genetics Resource Exchange \citep[AGRE,][]{bucan2009genome} which contain three parallel LRR sequences  of father-mother-offspring trio. Panel (a) of Figure \ref{fg::illustration} depicts the LRR sequence of mother's whole genome and clearly illustrates that it is impossible to pick CNV out by eyeballs. Panels (b) and (c) show a zoom-in plot of one of the detected CNVs from mother's sequence and a histogram of sizes (in terms of the number of biomarkers) of CNVs which are commonly detected by several different methods considered in this paper, respectively. The size of the entire LRR sequence in (a) is 561,466, while one of the depicted CNV in (b) has size of 6. We would like to remark that the most of CNVs are very short (as shown in (c)) and hidden in a long and noisy sequence, which makes the problem non-trivial. We will revisit this data in Section \ref{s:real} where a complete analysis is illustrated.

\begin{figure}[h]
\begin{center}
\subfigure[Mother's LRR Sequence]{
\includegraphics[width = 0.31\textwidth]{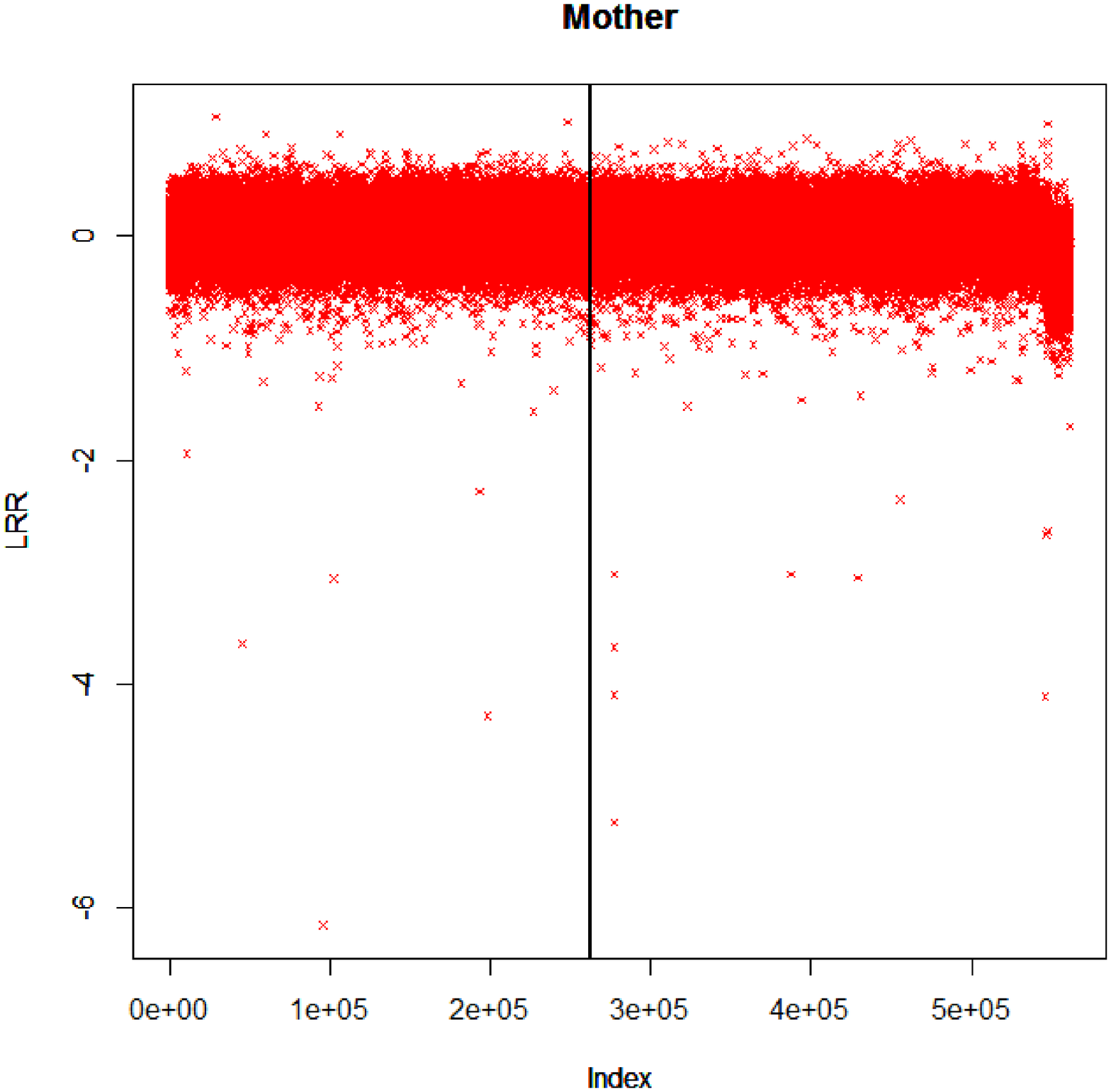}}
\subfigure[A zoom-in plot of one CNV from mother]{
\includegraphics[width = 0.31\textwidth]{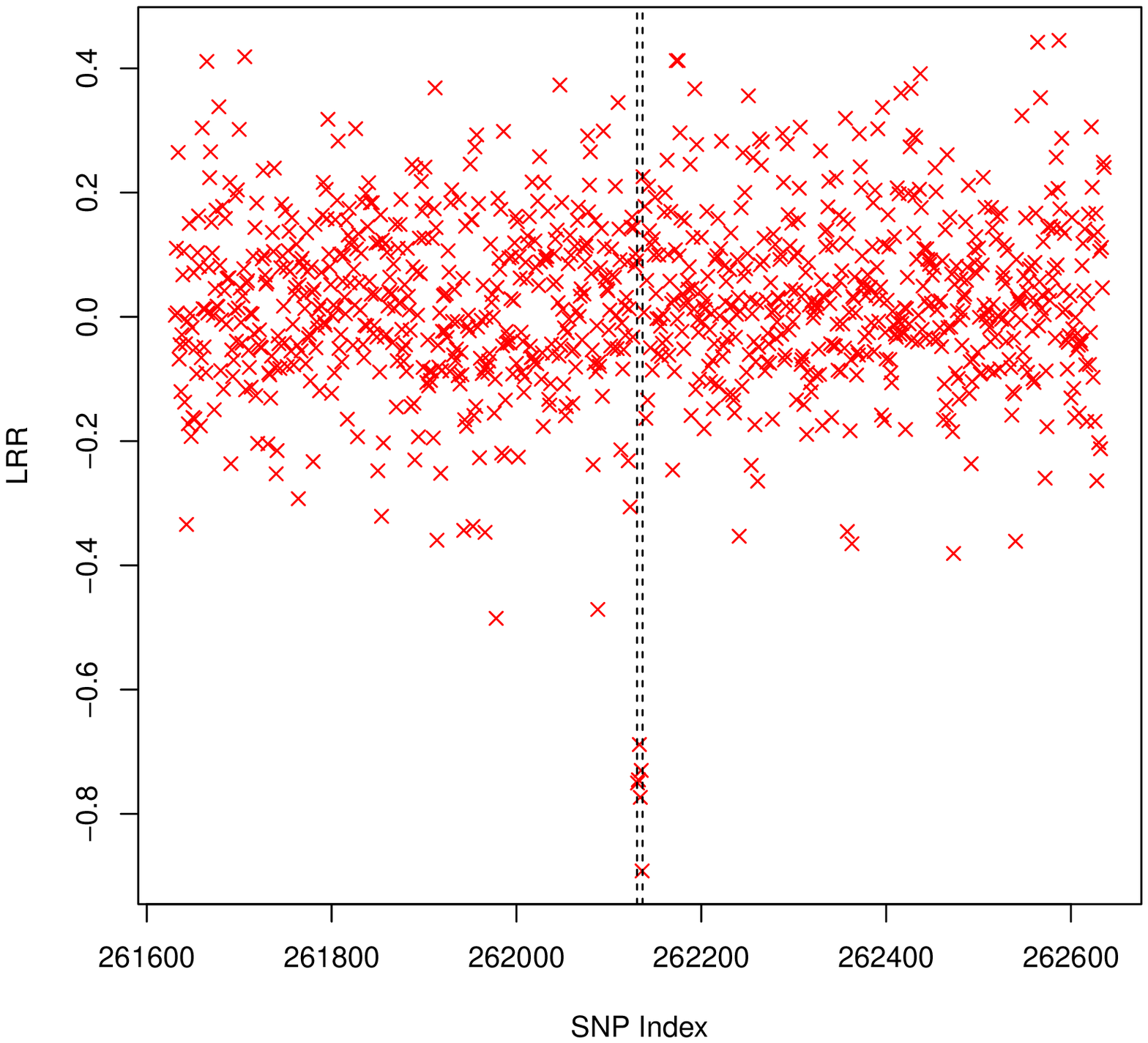}}
\subfigure[Sizes of CNVs detected]{
\includegraphics[width = 0.31\textwidth]{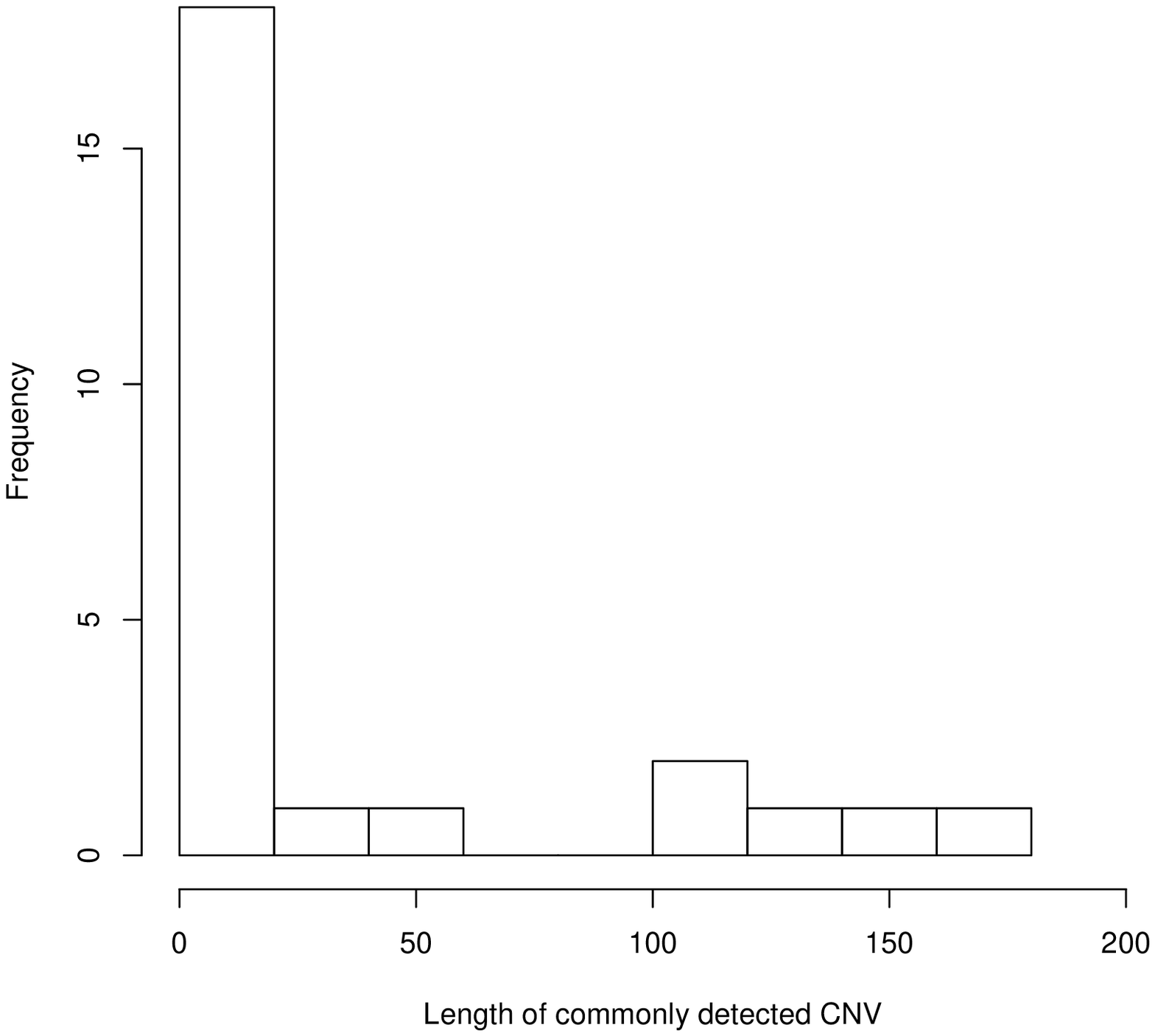}}
\end{center}
\caption{The subpanel (a) shows the whole sequence of the mother's LRR; (b) depicts one of them from the mother's LRR sequence of the AGRE trio SNP array data; and (c) is a histogram of length of CNVs commonly detected by several methods. } \label{fg::illustration}
\end{figure} 

In this context, a desirable  CNV detection method should be not only accurate enough to detect such short CNVs but also computationally fast enough to get estimates within a practically manageable time limit even for very long data sequence. Toward this, we propose a new change-point detection method called backward detection whose name comes from the backward variable selection in the linear regression. The backward detection is computationally efficient, with a complexity $O(n\log n)$ to analyze a sequence with length $n$. Moreover, it performs very well to pick out change-points which are very closely located. Therefore, it is an ideal tool for CNV detection. The idea of the backward detection is conceptually similar to the Wald's agglomerative clustering \citep{wald1963}, but different in that the location information of the sequence data is naturally employed. We also note that our method can be viewed as a ``bottom-up'' strategy for the change-point detection which has not studied extensively in the literature compared to ``top-down'' ones (such as binary segmentation) mainly due to its computational intensity. Recently, \cite{fryzlewicz2018} proposed an efficient ``bottom-up'' method for the general multiple change-point detection problem by using what they call the tail-greedy unbalanced Haar (TGUH) transformation. Yet, our numerical simulation shows that it suffers from detecting short and sparse signals which is common in CNV detection.

The rest of article is organized as follows. In Section \ref{s:model}, a normal mean change-point model and several popular detection strategies are introduced. In Section \ref{s:backward}, the backward detection is proposed in great details. A stopping rule for the backward detection is developed in Section \ref{s:cutoff}. Numerical performance of the proposed method is evaluated in Section \ref{s:sim}, and illustrations to both log-R ratio data from SNP array and the read depths from the aligned sequence data are given in Section \ref{s:real}. Final discussion follows in Section \ref{s:discussion}.

\section{Change-point Model} \label{s:model}
A normal mean change-point model assumes
\begin{align}\label{model1}
Y_i=\mu_i+\epsilon_i,\qquad i=1, 2, \cdots, n
\end{align}
with $\epsilon_i$ being iid $N(0, \sigma^2)$. The means $\mu_i$ are assumed to be piecewise constant with $K$ change points at $\bt = (t_1, t_2, \cdots, t_K)^T$. Denote $t_0=0$ and $t_{K+1}=n$ by convenience. It means that $\mu_i=\mu_j$ for any $i,j\in\{t_k+1, \cdots, t_{k+1}\}$, $0\leq k\leq K$, and $\mu_{t_k}\ne\mu_{t_k+1}$ for $1\leq k\leq K$. Yet the number of change points $K$ is typically unknown. The goal of change-point detection is to estimate both the number $K$ and location vector of change points $\bt$. Thus CNV detection can be regarded as a direct application of the change-point model (\ref{model1}). However, the CNVs are often very short and buried in a very long data sequence, which makes the problem even more challenging due to the high-dimensionality of $\bmu = (\mu_1, \cdots, \mu_n)^T$. The normal mean change-point model (\ref{model1}) is often reasonable in CNV application due to the random noise during the experiments \citep{barnes2008robust}.
A certain transformation can be considered otherwise. For example, the raw RD data are  discrete and spatially correlated due to the complicated sequencing process, and hence the normal error assumption is not proper. Nevertheless, the local median transformation can be used to ensure its normality \citep{tony2012robust}.

Suppose that the number of change points, $K$, is known. Then the change-point detection problem can be formulated in terms of minimizing the sum of squared errors (SSE). For a set of numbers $Y_i$ with index $i$ in a set $\cG$, we define their SSE  by $SSE(\{Y_i, i\in \cG\})=\sum_{i\in\cG}(Y_i-\bar Y_{\cG})^2$ where $\bar Y_{\cG}=\frac{1}{|\cG|}\sum_{i\in \cG} Y_i$ and $|\cG|$ denotes the cardinality of the index set $\cG$. Then change-point detection is to estimate $\bt$ by solving the following optimization problem
\begin{eqnarray}
\min_{\bt}&&\sum_{j=0}^{K} SSE\left(\left\{Y_i:  t_{j}+1\leq i\leq t_{j+1}\right\}\right),\label{obj:sse}\\
\mbox{subject to}&& 0=t_0<t_1<t_2< \cdots < t_K<t_{K+1}=n.\nonumber
\end{eqnarray}

Note that (\ref{obj:sse}) is inherently a combinatorial problem and very challenging for large $n$.  The total number of different combinations is $\frac{n!}{K!(n-K)!}\geq \left(\frac{n}{K}\right)^K$ which can be huge especially as in the application of CNV. This makes it difficult to detect change points by solving (\ref{obj:sse}) directly not to mention the fact that $K$ is typically not known. When $n$ is small and $K$ is bounded, the exhaustive search method has been studied by \cite{Yao1988} and \cite{Yao:1989}. They showed that an exhaustive search with BIC is consistent to estimate $K$ and $\bt$.
In order to improve computational efficiency, dynamic programming can be applied to solve (\ref{obj:sse}) with complexity of $O(n^2)$ \citep{BBM:2000,jackson2005algorithm}. \cite{killick2012optimal} developed an efficient algorithm named PELT that solves the problem with linear cost $O(n)$, but requires additional assumptions which might not be practical in certain applications. In general, these methods have not been widely applied in the CNV applications.

We remark that the mean change-point model (\ref{model1}) can be equivalently reformulated as a linear regression model \citep{HuangEtal:2005,fusedlasso:2008} and the change-point detection problem is then viewed as a variable selection one. Motivated by backward elimination methods in variable selection, we propose a stepwise procedure called backward detection (BWD) to solve the change-point detection problem. Some stepwise methods in the context of change-point detection have been explored. For example, a classical binary segmentation \citep[BS,][]{Vostrikova:1981} method applies single change-point detection tool recursively, identifying one change point each time, until that some stopping criterion is met. We refer BS as a forward detection method because it starts with a null model with no change point and sequentially detects change points, which is analogous to the forward variable selection in the regression context. In spite of its simplicity, as pointed out by \cite{OVLW:2004}, the forward detection is not able to to detect short segments buried in a long sequence of observations, which limits its utilization in certain applications such as CNV detection. The CBS \citep{OVLW:2004,VO:2007} modifies the BS by identifying two change points simultaneously and gained great popularity in the CNV detection. However, we observe from limited numerical studies that the CBS is still unsatisfactory when the true segment (i.e., CNV) is very short due to its nature from the forward detection. Moreover, CBS has higher computational complexity than the BS, which brings additional burden dealing with big data.

\section{Method} \label{s:backward}

\subsection{Why Not Forward Detection?}
In what follows, we elaborate why the forward detection may fail to detect short signals, which provides a clear motivation of the proposed BWD in CNV detection. The forward detection starts with no change point and try to detect the very first one by solving the following optimization problem
\begin{eqnarray} \label{eq::forward}
\min_{s_1}&&\sum_{j=0}^{1} SSE(\{Y_i, i\in\{s_{j}+1, \cdots, s_{j+1}\}\}),\label{obj:sse:fwd1}\\
\mbox{subject to}&& 0=s_0<s_1<s_2=n.\nonumber
\end{eqnarray}
The optimizer $\hat s_1$ of (\ref{obj:sse:fwd1}) estimates one of the change points and divides the data into two parts $\{Y_i, i\in\{1, 2, \cdots, \hat s_1\}\}$ and $\{Y_i, i\in\{\hat s_1+1, \cdots, n\}\}$. We may apply (\ref{obj:sse:fwd1}) to each of these two parts to detect further change points and this can be continued until we have identified all change points.

Note that the total number of combinations for (\ref{obj:sse:fwd1}) is $n$ and thus the corresponding optimization is feasible. However, as aforementioned the performance of the forward detection may not be satisfactory in some situation. For example, if there are only two change points at $t_1$ and $t_2$ and $\mu_i=0$ if $i < t_1$ or $i \ge t_2$  and $\mu$ if $t_1 \le  i < t_2$, the forward detection does not work well if the length of signal $L = t_2 - t_1$ is small while both $t_1$ and $n-t_2$ are large. However, this type of challenging situation is very common in the CNV applications as shown in Figure \ref{fg::illustration}.

In order to illustrate drawbacks of the forward detection, we consider a simplified scenario in which the locations of the two potential change points $t_j,$ $j=1,2$ are known, but it is not clear whether the associated mean $\mu$ is actually changed (i.e., $\mu = 0$ or not). The following Proposition 1 formally states that the forward detection asymptotically fails even in this simple scenario unless $L$ is sufficiently large compared to $n$. 

\begin{proposition}{Proposition 1.}{}
Suppose $\lim_{n \to \infty} t_1/n = c \in (0, 1)$ and $L = t_2 - t_1 = O(n^{\beta})$ for some $\beta \in [0, 1]$. If $\beta < 1/2$ then the forward detection fails as $n \to \infty$ for an arbitrary pair of $(\mu, \sigma^2)$ given.
\end{proposition}

\begin{proof}{Proof}
At the first step of the forward detection, it declares that the mean-change occurs at $t_j, j = 1, 2$ if $|\bar D_{n,t_j}| = | \bar Y_{t_j} - \bar Y_{n-t_j} |$ is significantly large enough. Here $\bar Y_t = \sum_{i=1}^{t}Y_i / t$ and $\bar Y_{n-t} = \sum_{i=t+1}^{n} Y_i / (n-t)$, for a given $t  \in \{1, \cdots, n-1\}$.

To test for $t_1$, the sampling distribution of $\bar D_{n,t_1}$ for a given $\mu$ is
\begin{align*}
\frac{\bar D_{n,t_1} + L\mu/(n-t_1)}{\hat \sigma_n \sqrt{\frac{1}{t_1} + \frac{1}{n-t_1}}} \quad \stackrel{\mathcal{D}}{\to} \quad N \left(0 ,1 \right),
\end{align*}
where $\hat \sigma_n^2$ denotes a consistent estimator of unknown $\sigma^2$. Then it can be shown that the associate asymptotic power converges to the nominal level $\alpha$ for any given pair of $(\mu, \sigma^2)$ if $\lim_{n \to \infty} n^{-1/2}L = 0$. Similar result can be shown for $t_2$  as well, which completes proof.
\end{proof}

Proposition 1 provides a necessary condition of the forward detection in terms of the relative length of the true signal length $L$ as a function of sample size $n$. The order of $L$, $\beta$, should be larger than $1/2$ for the original change-point model where the change points $t_1$ and $t_2$ are unknown. Recently \cite{fryzlewicz2013wild} showed that the forward selection is consistent to recover true change points when $\beta$ is larger than $3/4$.

\subsection{Backward Detection}
Contrary to the forward detection, the BWD starts from the opposite extreme that every single position is assumed to be a change point. Namely we begin with $n$ groups corresponding to these $n-1$ change points and each group contains only one observation. We introduce notations $\mathbb{G}=\{\cG_1, \cG_2, \cdots, \cG_n\}$ with $\cG_i=\{i\}$.

The BWD works by repeatedly merging two neighboring groups into one. Note that the merging of two neighboring groups will increase the total sum of squared errors.  For any two neighboring groups, we use the rise in the SSE to quantify the potential of merging them together. At each merging step, we choose to merge two neighboring groups with the smallest rise of SSE.   Namely we define
\begin{align} \label{eq::R}
R_{i}=SSE(\{Y_j, j\in \cG_i\cup \cG_{i+1}\})-SSE(\{Y_j, j\in \cG_i\})-SSE(\{Y_j, j\in  \cG_{i+1}\}),
\end{align}
where $SSE(\{Y_j, j\in  \cG\})$ denotes the sum of squared errors for all observations with indices in $\cG$.

At the beginning of iteration $m=0, 1, \cdots, n-2$, there are  $n-m$ groups. Denote the current groups by $\mathbb{G}^{(m)}=\{\cG_1^{(m)}, \cG_2^{(m)}, \cdots, \cG_{n-m}^{(m)}\}$ and the corresponding potential of merging two neighboring groups by $\{R_{1}^{(m)},R_{2}^{(m)},\cdots, R_{n-m-1}^{(m)}\}$. The superscript is used to represent the $m$th iteration.
Identify $j=\displaystyle\argmin_{i=1, 2, \cdots, n-m-1 } R_{i}^{(m)}$. Then we merge groups $\cG_j^{(m)}$ and $\cG_{j+1}^{(m)}$ into a new group. Updated grouping is denoted by $\mathbb{G}^{(m+1)}=\{\cG_1^{(m)}, \cG_2^{(m)}, \cdots, \cG_{j-1}^{(m)},  \cG_{j}^{(m)}\cup \cG_{j+1}^{(m)}, \cG_{j+2}^{(m)}, \cdots, \cG_{n-m}^{(m)} \}$ and potentials of merging is updated by $\{R_{1}^{(m)},R_{2}^{(m)},\cdots, R_{j-2}^{(m)}, R_{j-}^{(m)}, R_{j+}^{(m)}, R_{j+2}^{(m)}, \cdots, R_{n-m-1}^{(m)}\}$, where
\begin{align*}
R_{j-}^{(m)} = SSE(\{Y_j, j \in \cG_{j-1}^{(m)} & \cup \cG_{j}^{(m)} \cup \cG_{j+1}^{(m)}\}) - \\
    & SSE(\{Y_j, j \in \cG_{j-1}^{(m)}\})-SSE(\{Y_j, j\in  \cG_{j}^{(m)}\cup \cG_{j+1}^{(m)}\}), \mbox{ and} \\
R_{j+}^{(m)} = SSE(\{Y_j, j \in \cG_{j}^{(m)} & \cup \cG_{j+1}^{(m)} \cup \cG_{j+2}^{(m)}\}) - \\
    & SSE(\{Y_j, j\in  \cG_{j}^{(m)}\cup \cG_{j+1}^{(m)}\})-SSE(\{Y_j, j\in \cG_{j+2}^{(m)}\}).
\end{align*}
Now, the steps described above are repeatedly applied until a desire stoping rule is satisfied. The associated stoping rule is discussed in the following section. If not terminated, only one group will be survived at the end of iteration $n-2$.

Despite their structural similarity, the BWD is substantially different from the froward detection, since the null and alternative hypotheses at each step are reverted. At each step, the BWD tests the equivalence between the two group means while the forward tests their difference. Therefore, the BWD tends to stay with more groups with smaller sizes unless there is strong evidence to merge some of them and hence is more powerful to detect short signals buried on a long sequence. We also remark that the BWD starts with solving a series of local problems each of which focuses on finding structural changes in a small part of the data and eventually towards to a single global problem at the end that employs the entire sequence. On the other hand, the forward detection starts as a global method and divided it into several local problems. This is one of the reasons why BWD is preferred to identify short signals in lengthy noise sequences, where the local methods are known to outperform the global methods in general.

Finally, the BWD algorithm can be summarized as follows.
\begin{itemize}
\item [] \verb"Input": $Y_1, \cdots, Y_n$. \par
\begin{itemize}
\item [1.] Initialize $\mathbb{G}^{(1)} = \{\{1\}, \cdots, \{n\}\}$ and $\bR^{(1)} = \{R_1^{(1)}, \cdots, R_{n-1}^{(1)}\}$ from (\ref{eq::R}).
\item [2.]At the $m$th iteration, $m = 1, \cdots, n-1$:
    \begin{itemize}
    \item [2--1] Obtain $j = \argmin_i R_i^{(m)}$.
    \item [2--2] Break the loop if $R_j^{(m)}$ is larger than a prespecified cutoff, and go to the next step otherwise.
    \item [2--3] Update
    \begin{align*}
    & \mathbb{G}^{(m+1)} =\{\cG_{1}^{(m)}, \cdots, \cG_{j-1}^{(m)},\cG_{j}^{(m)}\cup\cG_{j+1}^{(m)},\cG_{j+2}^{(m)},\cdots,\cG_{n-m}^{(m)}\},\\
    & \bR^{(m+1)}  = \{R_{1}^{(m)},R_{2}^{(m)},\cdots, R_{j-2}^{(m)}, R_{j-}^{(m)}, R_{j+}^{(m)}, R_{j+2}^{(m)}, \cdots, R_{n-m-1}^{(m)}\}, \\
    & K = n - m -1.
    \end{align*}
\end{itemize}
\end{itemize}
\item [] \verb"Output": $\mathbb{G}^{(K)}.$
\end{itemize}

\subsection{Modification for epidemic change-points}
In CNV analysis, most parts of sequence (normal) have a known baseline mean, say $\mu_0$, and a mean-change away from $\mu_0$ (variant) is followed by a change back to $\mu_0$. This is often referred to as the epidemic change-points \citep{yao1993tests} and an important feature of CNV analysis. In order to to take into account such pairing structure, we modify the algorithm by adding the following Step between Step 2--2 and 2--3.

\begin{itemize}
\item If the sample average of observations in the merged sets, $\bar Y_{\cG_{j}^{(m)}\cup\cG_{j+1}^{(m)}}$ is not significantly different from the baseline mean $\mu_0$, i.e., $ \sqrt{v}{\left| \bar Y_{\cG_{j}^{(m)}\cup\cG_{j+1}^{(m)}} - \mu_0 \right|}/\hat \sigma_n > z_\alpha$ where $z_\alpha$ is the upper $\alpha$th quantile of standard normal random variable and $v = \left|\cG_{j}^{(m)}\cup\cG_{j+1}^{(m)}\right|$, then update $R_{j-}^{(m)}$ and $R_{j+}^{(m)}$ based on $\mu_0$ instead of $\bar Y_{\cG_{j}^{(m)}\cup\cG_{j+1}^{(m)}}$.
\end{itemize}

Finally, we develop \texttt{bwd} R-package that implements the proposed algorithm, and available on CRAN.

\subsection{Computational Complexity}
Computational efficiency is of  practical interest in CNV applications due to their inherent  high-dimensionality. At each iteration in the BWD, the most computationally intensive part is to find $j = \argmin_i R_i^{(m)}$ which takes $O(n)$ at the worst. This gives the total complexity of $O(n^2)$, which is too slow especially when $n$ is very large.

However, it is realized that finding maximum and corresponding index is straight forward once $\bR^{(1)}$ is ordered, which takes $O(n \log n)$ computations. Notice that the sorting step is required onetime at the initial stage. Once it is sorted, it take $O(1)$ to find the maximizer index at the $m$th iteration while we need additional effort to update $\bR^{(m+1)}$ in an ordered fashion. However, such an update takes only $O(\log n)$ computations. In particular, we borrow the idea from the bi-section method, a well-known root finding algorithm. We can first compare $R_{j+}^{(m)}$ (or $R_{j-}^{(m)}$) to the median of the values in $\bR^{(m)}$. Compare it the 75\%th percentile if it is greater than the median and 25\%th percentile otherwise. We continue this until finding its exact location. Finally, the total computational complexity of the BWD is then reduced to $O(n \log n)$.

\section{Stopping Rule}\label{s:cutoff}

In every step of the BWD, two small groups are merged into a bigger group and we want to test whether this merging removes a real change point. In such a standard case, it is natural to use $t$-statistic. Since the unknown variance is homogeneous across all the observations, a global estimate of the noise variance is used at every step. At the $m$th iteration the following statistics $S_{(m)}$ is used to determine when to stop the procedure where
\begin{equation} \label{eq::t.stat}
S_{(m)} = \frac{\left| \bar Y_j^{(m)}-\bar Y_{j+1}^{(m)} \right|}{\hat \sigma_n \sqrt{\big|\cG_j^{(m)}\big|^{-1} + \big|\cG_{j+1}^{(m)}\big|^{-1}}},
\end{equation}
and the backward procedure is terminated if $S_{(m)}$ is too large. Here $\hat \sigma_n^2$ denotes an estimate of the unknown noise variance based on all the observation. If the true signals are very short and sparse, the sample variance of $Y_i$ can also be used in practice as a simple alternative. The use of global estimate brings an additional saving in computation since $R_j = \hat \sigma_n^2 S_{(m)}^2$. We use $\hat \sigma_n^2 = n^{-1} \sum_{i=1}^{n} \big(Y_i - \bar Y_{i}^{(h)}\big)^2$ with $\bar Y_{i}^{(h)} = \sum_{j = i-h}^{i+h} Y_j/(2h+1)$ for a given window $h > 0$ in the upcoming analysis as used in \cite{NiuZhang:2010}. An alternative estimate is the median absolute deviation (MAD) estimator, as pointed out in \cite{CaiJengLi:2010}.

Similar to usual t-statistic, (\ref{eq::t.stat}) may cause a false alarm when both of the two groups are of small sizes. To avoid the possible false alarm caused by small group sizes, we can set $S_{(m)} = 0$ if both of the two consecutive segments are shorter than a minimum number $M$. The $M$ can be chosen to be, say, 3 or 5 depending on applications. Such modification is acceptable in CNV application since it is very unlikely that any of two CNVs are closely located.

Now, the question is how large the critical value should be  to attain a desired target level $\alpha$ where $\alpha$ denotes the familywise error rate (FER) of the proposed BWD.
We remark that $(1-\alpha/2)$th quantile of $t$-statistic with the associated degrees of freedom will fail since $S_{(m)}$ is \textbf{correlated other group means} via the maximizer index $j$.
We propose the following numerical procedure to select a cutoff value that controls FER being at most $\alpha$.
\begin{enumerate}
\item Repeat the steps (a) -- (c) below $B$ times: for each of the $b$th iteration, $b = 1, \cdots, B$,
\begin{enumerate}
\item [1-(a)] Randomly generate a sequence of size $n$ from the null distribution that there exists no change point.
\item [1-(b)] Apply the backward procedure until merging the whole sequence into one group.
\item [1-(c)] $u_b = \max_{m = 1, \cdots, n-1} S_{(m)}$.
\end{enumerate}
\item The ($1- \alpha$)th sample quantile of $u_1, \cdots, u_B$ would be the cutoff value which attains a given level $\alpha$.
\end{enumerate}
We remark that the cutoff value is chosen from the null distribution of $\max_{m = 1, \cdots, n-1} S_{(m)}$, \textbf{not $S_{(m)}$ (step 1-(c))}, thus $\alpha$ controls the FER. The very first step 1-(a) that simulates samples from the null distribution is crucial in the proposed numerical procedure. Toward this we consider two scenarios: i) normality is assumed to be true while a noise variance $\sigma^2$ is still unknown. ii) neither normality nor $\sigma^2$ are known. In the first scenario, we can generate samples form the standard normal distribution. Notice also that this can be easily extended to any distribution other than normal distribution, whatever it is known. In the second scenario when underlying distribution is not known, the null distribution can be obtained by randomly permuting or bootstrapping residuals $r_i = Y_i - \bar Y_{i}^{(h)}, i = 1, \cdots, n$.

\begin{figure}[!htbp]
\begin{center}
\subfigure[cutoff vs $n$]{
\includegraphics[width = 0.48\textwidth]{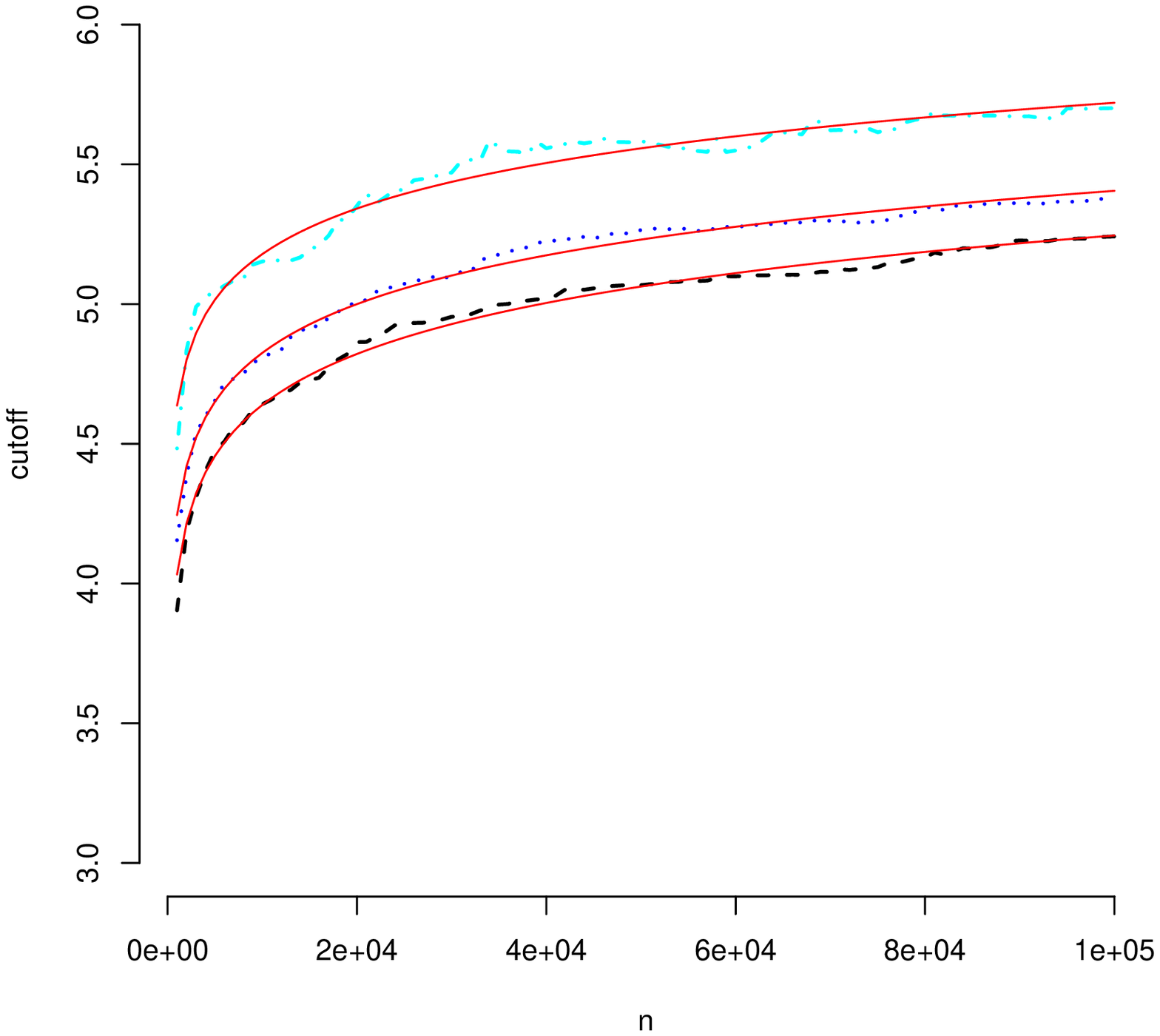}}
\subfigure[cutoff vs $\log n$]{
\includegraphics[width = 0.48\textwidth]{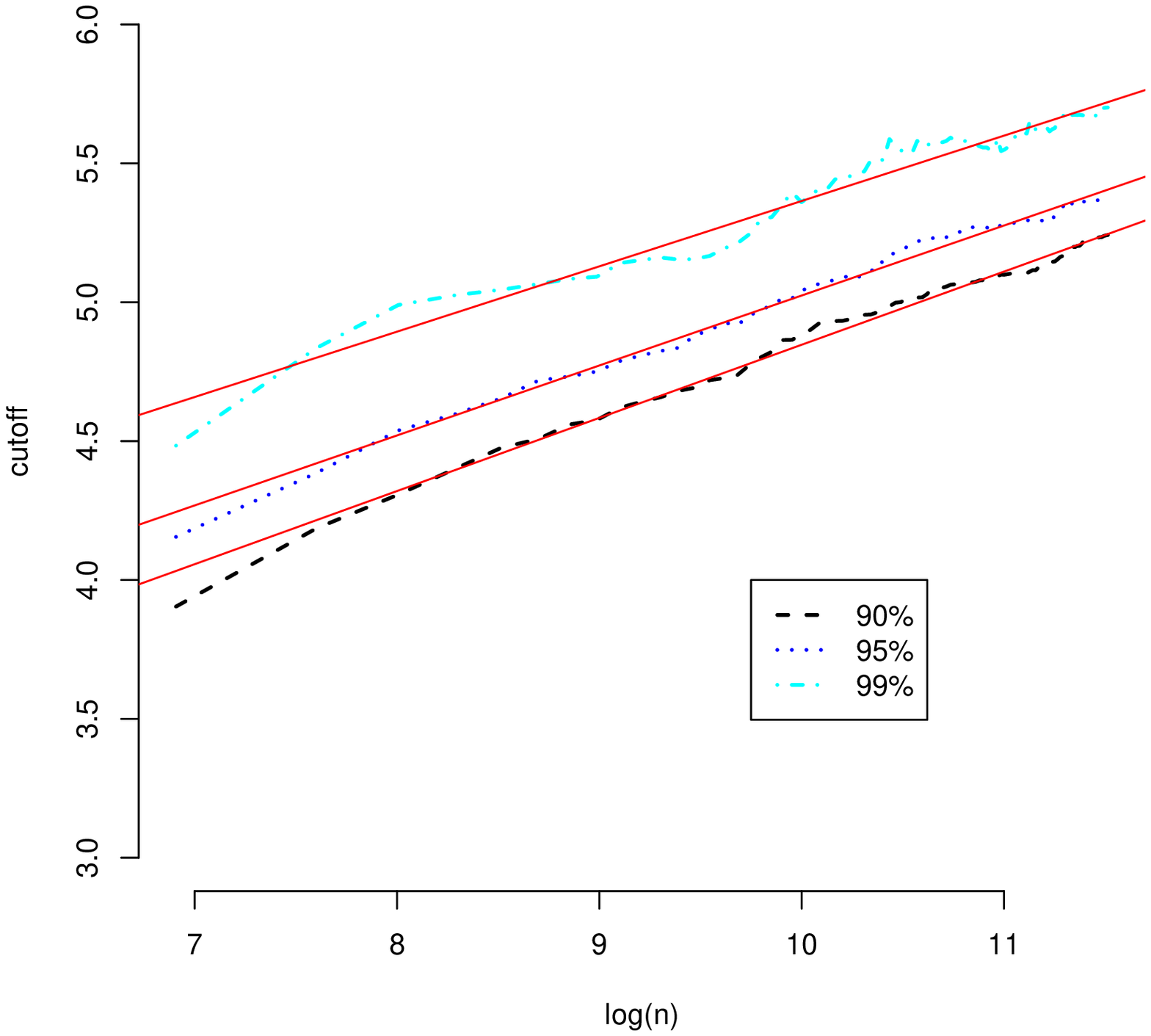}}
\end{center}
\caption{Log-linear relationship between estimated cutoffs ($\alpha = 0.01, 0.05, 0.10$) and sample size $n$ under normality assumption. The (red) solid lines in (a) and (b) are fitted regression lines of cutoffs on sample size $n$ and $\log n$, respectively. } \label{fg::cutoff}
\end{figure}

The proposed numerical procedure becomes computationally too intensive especially when sample size is very large, for instance over a million, which is not uncommon in CNV applications. Under normality assumption, we numerically investigate cutoff values for different $\alpha = 0.01, 0.05$, and $0.10$ as functions of sample size $n$. Figure \ref{fg::cutoff} depicts estimated cutoffs for different sample sizes from 1,000 to 100,000 by 1,000 and it shows clear log-linear relationship between the estimated cutoffs and the sample size $n$. Thus desired cutoffs for large $n$ can be approximated from the fitted regression line.

\section{Simulated Examples} \label{s:sim}

We evaluate the performance of the proposed backward procedure via numerical comparison against existing methods. The target levels considered are $\alpha = .01$ and $.05$. We consider both of the original BWD (BWD1) and the modified BWD (BWD2) for epidemic change-points under the assumption that the baseline mean $\mu_0$ is known. As described in Section \ref{s:cutoff}, there are two ways to obtain the cutoff values depending on how to simulate null samples. We can obtain a cutoff either from standard normal samples under normality assumption (cutoff1) or from the permuted residuals if normality assumption is not valid (cutoff2). We take the former in Section 5.1 with Gaussian error and the latter in Section 5.2 with non-Gaussian error.
We consider the CBS, WBS, LRS, TVP as competing methods. The CBS is one of the most widely used methods in the literature and shares similar principals of a typical stepwise approach with the proposed method. The WBS is a recent development based on binary segmentation (i.e., forward detection) that overcomes its shortcoming when detecting short signal. The LRS is carefully designed method for detecting sparse and short signals and known as optimal under some required model assumptions that include normality, and shortness and sparsity of the signals. In addition, we compare our method to a recently proposed bottom-up method, called TGUH \citep{fryzlewicz2018}. TGUH is designed from general change-point detection problem and our simulation show that it suffers from detecting short signals.

We consider the following mean change model
\begin{align*}
y_i = \sum_{k=1}^\kappa \delta_k \bone_{\{i \in I_k\}} + \epsilon_i
\end{align*}
where $\kappa$ denotes the number of signal segments (i.e., CNVs), $\delta_k, k = 1, \cdots, \kappa$, are unknown means of true signals, and $I_k \cap I_{k^{\prime}} = 0$ for any $k \neq k^\prime$.  We set $|I_k| = L$ and $\delta_k = \delta, k = 1, \cdots, \kappa$, and hence the strength of true signals are controlled by $L$ and $\delta/\sigma$. We consider two different noise distributions of $\epsilon$ including normal distribution and $t$-distribution with degrees of freedoms of $df$. We set $(n, L, \delta) = \{1000, 3000, 5000\} \times \{5, 10\} \times \{1.5, 2.0, 2.5\} $ with $\sigma = 1$ for normal model, and $(L, df) = \{5, 10\} \times \{10, 5\}$ with $n = 1000$ and $\delta = 3$ for $t$-distribution model. The number of true segments is given by $\kappa = n/1000$ and minimum distance between two true segments is set to 200. Numerical performance is evaluated based on 1,000 independent repetitions.

We claim that the signal segment $I_k$ is correctly detected by $\hat I_k$ if $I_k \cap \hat I_k \neq \phi$ and $|\hat I_k| < 2L$. In order to measure performance of the methods the following two measures are considered.
\begin{itemize}
\item [-] Sensitivity: (\# of correctly detected signals) / (\# of true signals, $\kappa$)
\item [-] Precision :  (\# of correctly detected signals) / (\# of detected signals)
\end{itemize}
Sensitivity relates to the ability to identify true signals and precision measures reliability of the detected signals. Notice that both measures lie between zero and one (by setting $0/0 = 0$) and a method is perfect if both measures have a value of one.

\subsection{Gaussian Error}

\begin{table}[t] \scriptsize
\begin{center}
\begin{tabular*}{\textwidth} {@{}c@{\extracolsep{\fill}}
                                 c@{\extracolsep{\fill}}
                                 c@{\extracolsep{\fill}}
                                 c@{\extracolsep{\fill}}
                                 r@{\extracolsep{\fill}}
                                 r@{\extracolsep{\fill}}
                                 r@{\extracolsep{\fill}}
                                 r@{\extracolsep{\fill}}
                                 r@{\extracolsep{\fill}}
                                 r@{\extracolsep{\fill}}
                                 r@{\extracolsep{\fill}}
                                 r@{\extracolsep{\fill}}
                                 r@{\extracolsep{\fill}}
                                 r@{\extracolsep{\fill}}
                                 r@{\extracolsep{\fill}}
                                 r@{}}
\hline
    &\multicolumn{3}{c}{$L$}      &\multicolumn{6}{c}{5} & \multicolumn{6}{c}{10} \\
    &\multicolumn{3}{c}{$\delta$} &\multicolumn{2}{c}{1.5} & \multicolumn{2}{c}{2.0} & \multicolumn{2}{c}{2.5} &
               \multicolumn{2}{c}{1.5} & \multicolumn{2}{c}{2.0} & \multicolumn{2}{c}{2.5} \\
$n$ &\multicolumn{3}{c}{Methods}& Sen. & Pre. & Sen. & Pre. & Sen. & Pre. & Sen. & Pre. & Sen. & Pre. & Sen. & Pre. \\ \hline
\multirow{6}{9mm}{1,000} &
\multirow{4}{6mm}{BWD}  &
  \multirow{2}{5mm}{.01}
     & cutoff1 & .195 	& .924 	& .581 	& .975 	& .919 	& .985 & .646 	& .983 	& .960 	& .993 	& .998 	& .994\\
&&   & cutoff2 & .229 	& .909 	& .640 	& .965 	& .925 	& .982 & .696 	& .971 	& .970 	& .983 	& .999 	& .985\\
&&\multirow{2}{5mm}{.05}
     & cutoff1 & .335 	& .819 	& .727 	& .910 	& .952 	& .933 & .777 	& .914 	& .983 	& .939 	& .999 	& .945\\
&&   & cutoff2 & .335 	& .819 	& .726 	& .913 	& .950 	& .943 & .773 	& .920 	& .982 	& .944 	& .999 	& .948\\
&\multicolumn{3}{c}{CBS}
      &  .165 	& .948 	& .555 	& .975 	& .900 	& .979 & .648 	& .972 	& .972 	& .976 	& .999 	& .977\\
&\multicolumn{3}{c}{WBS}
      &   .176 & .884	& .570	& .942	& .909	& .954 & .634 	& .948 	& .971 	& .956 	& .999 	& .957\\
&\multicolumn{3}{c}{LRS}
      & .216 	& .911 	& .638 	& .973 	& .942 	& .983 & .701 	& .979 	& .984 	& .988 	&1.000 	& .989\\ \hline 

\multirow{6}{9mm}{3,000} &
\multirow{4}{6mm}{BWD}  &
  \multirow{2}{5mm}{.01}
     & cutoff1 & .123 	& .966 	& .493 	& .990 	& .856 	& .995 & .543 	& .992 	& .954 	& .997 	& .999 	& .999 \\
&&   & cutoff2 & .161 	& .947 	& .578 	& .985 	& .902 	& .992 & .610 	& .988 	& .965 	& .995 	& .999 	& .997 \\
&&\multirow{2}{5mm}{.05}
     & cutoff1 & .229 	& .911 	& .653 	& .972 	& .926 	& .982 & .700 	& .972 	& .980 	& .982 	& .999 	& .988 \\
&&   & cutoff2 & .250 	& .900 	& .676 	& .964 	& .931 	& .976 & .718 	& .968 	& .982 	& .978 	&1.000 	& .981 \\
&\multicolumn{3}{c}{CBS}
               & .103 	& .960 	& .490 	& .984 	& .882 	& .985 & .572 	& .983 	& .974 	& .982 	& .999 	& .983 \\
&\multicolumn{3}{c}{WBS}
         	   & .079  	& .967 	& .399 	& .991 	& .819  & .990 & .470 	& .991 	& .938 	& .992 	& .999  & .990\\
&\multicolumn{3}{c}{LRS}
               & .152 	& .948 	& .553 	& .986 	& .898 	& .993 & .606 	& .988 	& .972 	& .994 	& .999 	& .996 \\ \hline 

\multirow{6}{9mm}{5,000} &
\multirow{4}{6mm}{BWD}  &
  \multirow{2}{5mm}{.01}
     & cutoff1 & .116 	& .967 	& .482 	& .994 	& .863 	& .997 & .532 	& .992 	& .945 	& .996 	& .999 	& .996\\
&&   & cutoff2 & .112 	& .971 	& .472 	& .995 	& .852 	& .998 & .517 	& .992 	& .941 	& .997 	& .999 	& .997\\
&&\multirow{2}{5mm}{.05}
     & cutoff1 & .212 	& .936 	& .632 	& .981 	& .924 	& .991 & .667 	& .981 	& .974 	& .989 	& .999 	& .991\\
&&   & cutoff2 & .213 	& .937 	& .623 	& .983 	& .917 	& .993 & .664 	& .982 	& .974 	& .989 	& .999 	& .991\\
&\multicolumn{3}{c}{CBS}
               & .088 	& .973 	& .468 	& .989 	& .890 	& .987 & .551 	& .983 	& .966 	& .985 	&1.000 	& .987\\
&\multicolumn{3}{c}{WBS}
               & .052  	& .985	& .327 	& .994 	& .730 	& .995 & .383  	& .994 	& .891 	& .995 	& .997 	& .995\\
&\multicolumn{3}{c}{LRS}
               & .130 	& .962 	& .512 	& .992 	& .881 	& .996 & .563 	& .992 	& .957 	& .997 	& .999 	& .997\\ \hline 
\end{tabular*}
\caption{Performances under normal error - The BWD and LRS outperform the CBS and WBS. The backward detection with $\alpha = .05$ ($.01$) shows higher (lower) sensitivity but lower (higher) precision compared to the LRS.} \label{tb::normal}
\end{center}
\end{table}

In many applications including CNV detection, normality assumption is often used. Although our backward procedure does not strongly require normality assumption, it performs best under normal noise due to the use of squared error loss. Table \ref{tb::normal} reports performance of different methods considered in various scenarios under normality. The both original and modified versions of BWD outperform others except LRS in most scenarios considered. This is because the true signal is very short ($L = 5$ and $10$). LRS performs quite well. This is not surprising since the designed simulation model perfectly satisfies assumptions required for the LRS. The modification for epidemic change-points is useful when the true signals are not strong. TVP is very fast, but gives too many false positives. TGUH shows very low sensitivity indicating that it cannot detect short signals well. It is interesting to observe that the BWD still performs comparably well in the sense that it outperforms the LRS in terms of sensitivity with $\alpha = 0.05$ and in terms of precision with $\alpha = 0.01$. It is another benefit of the BWD that it controls relative importance between sensitivity and precision (or specificity) through the target level $\alpha$. We also remark that the BWD is simple and does not require stringent model assumptions such as sparsity.

\subsection{Non-Gaussian Error}

Performance of the change-point detection methods are evaluated under t-distributed noise. The BWD does not require normality assumption and hence not overly sensitive to the violation of normality, whereas LRS does. Before applying the LRS, we standardized the observations first by using sample mean and sample standard deviation. Remark that such naive estimates would work fairly well for standardizing observation since the signals are very short compared to the entire sequence of data. Table \ref{tb::t} contains numerical performance of the methods under consideration. Advantages of the BWD are much clearer than the previous setting with normality. The CBS fails to detect true signals when the signal strength is not very strong while the backward procedure performs well in all the scenarios considered. Both of the WBS and LRS are good in terms of sensitivity but they detect too many false signals in this case. Again, BWD2 outperforms BWD1 when the true signals are not strong.

\begin{table}[h]
\begin{center}
\caption{Performances under t-distributed error - The BWD outperforms all other methods for detecting short signals. } \label{tb::t}
\begin{tabular*}{30pc}{@{\hskip5pt}@{\extracolsep{\fill}}r@{}l@{}r@{}r@{}r@{}r@{}r@{}r@{}r@{}r@{}@{\hskip5pt}}
\toprule
    \multicolumn{2}{c}{$L$}      &\multicolumn{4}{c}{5} & \multicolumn{4}{c}{10}\\
    \multicolumn{2}{c}{$df$} &\multicolumn{2}{c}{10} & \multicolumn{2}{c}{5} &
                              \multicolumn{2}{c}{10} & \multicolumn{2}{c}{5} \\
\multicolumn{2}{c}{Methods}& Sen. & Pre. & Sen. & Pre. & Sen. & Pre. & Sen. & Pre.\\ \hline
\multirow{2}*{BWD1}
     & .01 & .932	& .990	& .589	& .982	& 1.000	& .993	& .968	& .983\\
     & .05 & .965	& .949	& .879	& .879	& 1.000	& .952	& .994	& .942\\
\multirow{2}*{BWD2}
     & .01 & .936	& .993	& .606	& .985	& 1.000	& .996	& .967	& .986\\
     & .05 & .969	& .952	& .876	& .893	& 1.000	& .958	& .993	& .954\\
\multicolumn{2}{c}{CBS}
               & .853 	&.987 	&.287 	&.986 &  .998 	&.986 	&.864 	&.990\\
\multicolumn{2}{c}{WBS}
               &.981    &.867   &.939 	&.644 & 1.000   &.865   &.999   &.651\\
\multicolumn{2}{c}{LRS}
               & .983 	&.711 	&.907 	&.338 & 1.000 	&.741 	&.999 	&.375\\ 
\multicolumn{2}{c}{TVP}
               & .997   &.183 	&.974   &.212 & 1.000   &.204 	&.999   &.230\\ 
\multicolumn{2}{c}{TGUH}
               & .273   &.720 	&.545   &.619 &  .391   &.752 	&.699   &.666\\ \hline
\end{tabular*}
\end{center}
\end{table}

\subsection{Empirical Test Level}

\begin{table}[h] 
\begin{center}
\caption{Estimated level - The cutoff1 performs well under normality assumption is true and the cutoff2 can be used if the normality assumption is suspicious.} \label{tb::level}
\begin{tabular*}{30pc}{@{\hskip5pt}@{\extracolsep{\fill}}c@{}c@{}c@{}c@{}c@{}c@{}c@{}c@{}@{\hskip5pt}} \hline
\multicolumn{2}{l}{} & &\multicolumn{3}{c}{Normal}            & t(10) & t(5) \\
\multicolumn{2}{l}{Methods} & $\alpha$ &    $n=1000$ &  $n=3000$ &	$n=5000$ &	  $n = 1000$ &   $n = 1000$ \\ \hline
\multirow{4}*{BWD}
&\multirow{2}*{cutoff1}
 &0.01   &	.011 &	.013 & 	.009 &	.032 &	.030 \\
&&0.05   &	.051 &	.046 & 	.058 &	.077 &	.065 \\ \cline{2-8}
&\multirow{2}*{cutoff2}
&0.01   &	.014 &	.018 & 	.010 &	.007 &	.002 \\
&&0.05   &	.065 &	.070 &	.053 &	.052 &	.022 \\ \hline
\multicolumn{2}{c}{CBS} &N/A&    .007&	.011& 	.010& 	.003&	.003 \\
\multicolumn{2}{c}{WBS} &N/A&    .016&	.005& 	.004& 	.112&	.395\\
\multicolumn{2}{c}{LRS} &N/A&    .022&	.028& 	.028& 	.388&	.905 \\
\multicolumn{2}{c}{TVP} &N/A&    .966&	.609&	.600&	.968&	.925 \\
\multicolumn{2}{c}{TGUH}&N/A&    .030&	.023& 	.019& 	.098&	.210 \\ \hline
\end{tabular*}
\end{center}
\end{table}

We numerically check whether the backward procedure actually attains a target nominal level $\alpha$ under the null hypothesis that there exists no signal. Since the two version of BWD show similar result we report the results of the original version only to avoid redundancy. We generate samples under the null hypothesis by letting $\delta = 0$ and report the proportion of the cases that any signal is detected by each methods (Table \ref{tb::level}). Recall that we have two scenarios. The `cutoff1' assumes normality and hence the levels are correct if the data are indeed from normal model but cannot satisfy nominal level if the data are from t-distribution. In this case, the `cutoff2' can be used as an alternative and the results seem good enough to be used in practice. Notice that there are a couple of cases in which `cutoff2' fails to attain the nominal level, which is partially due to  the uncertainty of the null distribution. The CBS seems very conservative to detect signal and both LRS and TGUH break down when the normality assumption is not valid. TVP fails again to control the type I error. 
As aforementioned, it is another distinguished advantage of the proposed BWD to be able to control type 1 error. This is practically attractive since the relative importance of sensitivity and specificity varies depending on applications.

\section{Real Data Illustration} \label{s:real}

\subsection{Trio Data from SNP array}

The (original) BWD is demonstrated for the SNP array data collected from the Autism Genetics Resource Exchange \citep[AGRE,][]{bucan2009genome}. The data set contains three parallel sequences of log $R$ ratio (LRR) for 547,458 SNPs over 23 chromosomes of father-mother-offspring trio.

\begin{figure}[!htbp]
\begin{center}
\subfigure[chromosome1]{
\includegraphics[height = 0.7\textwidth, angle = 270]{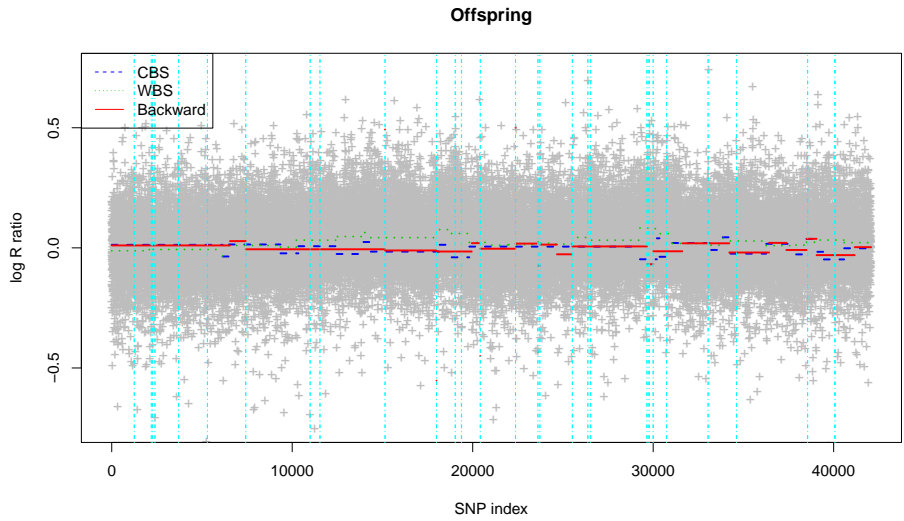}}
\subfigure[chromosome2]{
\includegraphics[height = 0.7\textwidth, angle = 270]{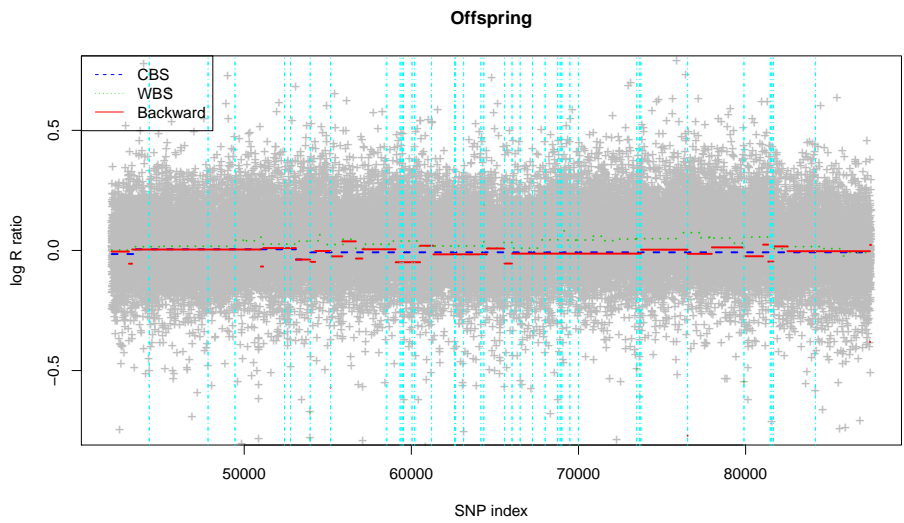}}
\caption{CNV detection results for chromosome 1 and 2 of offspring - All of the CBS, WBS, and BWD show quite different results. The vertical lines represent the signals detected by LRS.} \label{fg::result}
\end{center}
\end{figure} 

All methods considered in Section \ref{s:sim} are applied except TVP and TGUH. For LRS, the data are standardized by the sample mean and variance. We set $\alpha = 0.05$ and the corresponding cutoff value is approximated from the log-linear relation between cutoff and sample size under normal assumption as described in Section \ref{s:cutoff}. We apply each of the methods in chromosome-wise and Figure \ref{fg::result} shows the results for the first two chromosomes (chromosomes 1 and 2) of the offspring. The three different types (and colors) of horizontal segments are estimated $\mu_t$ by the CBS, LRS and BWD, respectively. Notice that the LRS only detects very short and sparse signals and the detected signals are marked as vertical lines. We would like to point out that although all of the CBS, WBS and BWD are developed under a similar framework, the results are quite different. For example in chromosome 2 (Subfigure (b)), the CBS detects no change point after around 54,000th SNP position, while both the BWD and WBS detect several.

\begin{figure}[h]
\centering
\includegraphics[width = 0.9\textwidth]{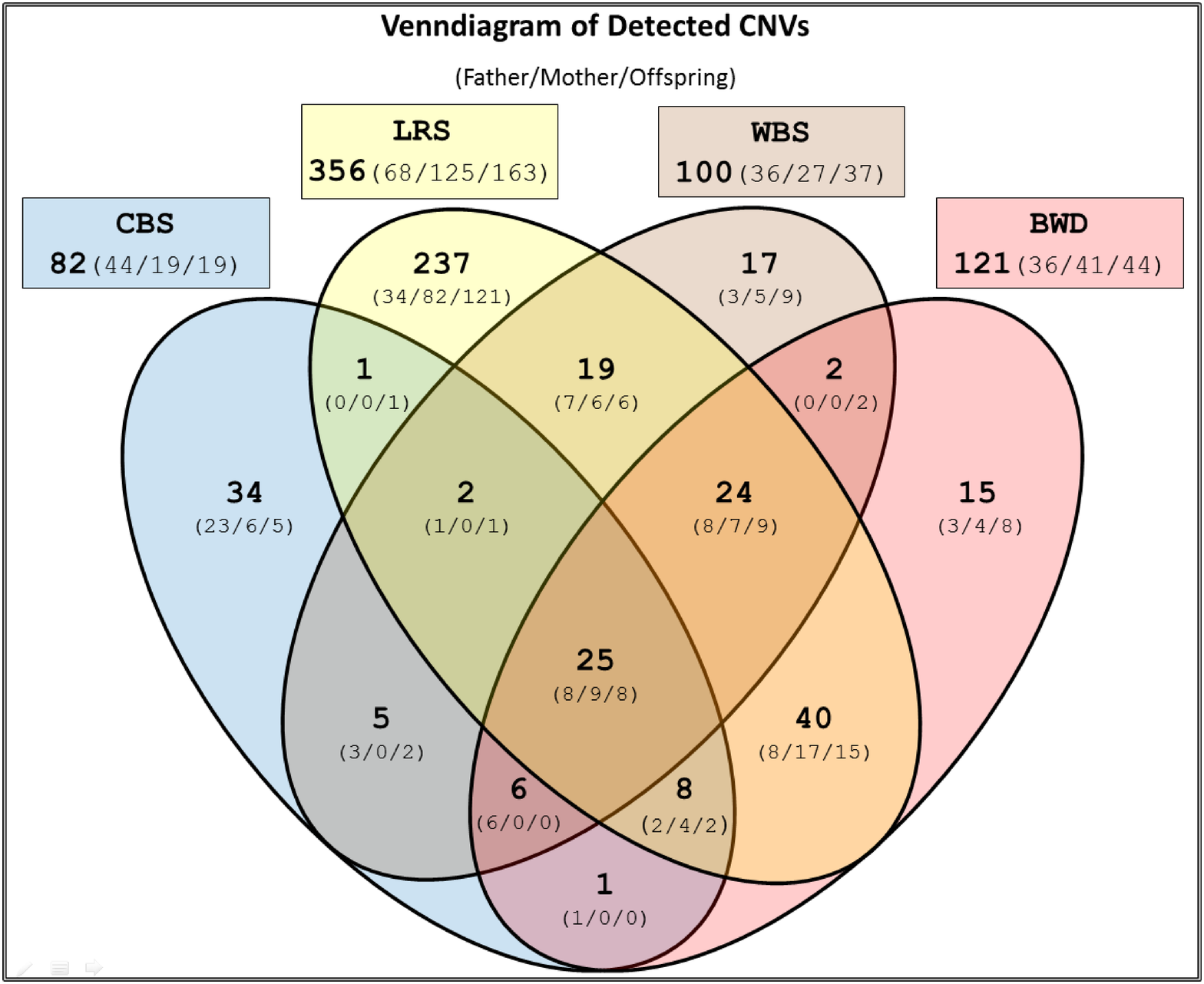}
\caption{Vendiagram of detected CNVs (father/mother/offspring) - The BWD calls the least number of unique CNVs most of which are likely to be false positives, while it misses only 2 CNVs that are identified by all other methods, while CBS, WBS, and LRS miss 24, 8, and 6 such CNVs, respectively.
}\label{fg::ven}
\end{figure} 

The complete CNV detection results for the trio data are summarized by a Venn diagram in Figure \ref{fg::ven} that reports the number of CNVs detected by different methods for each of the trio (father/mother/offspring) as well as the collapsed.
We consider  short detected segments  whose lengths are between 2 and 200 in terms of SNP index as CNVs.
First, the LRS detects a much larger number of CNVs than other competing methods, while majorities ($237/356 = 66.5\%$) are unique calls which are likely to be suspicious as false signals. The BWD calls 121 CNVs which is larger than those from the CBS (84) and WBS(100), while unique calls by the BWD is only 12.4\% (15/121) smaller than any others (CBS: 34/82 = 41\%; WBS: 17/100 = 17\%). This can be interpreted that the BWD shows best precision if we assume that the most of CNVs uniquely called by a single method are false positives. Next, BWD misses only 2 CNVs that are identified by all other methods, while CBS, WBS, and LRS miss 24, 8, and 6 such CNVs, respectively, meaning that the BWD outperforms others in terms of sensitivity as well.
Finally, 25 CNVs identified by all the methods can be regarded as true CNVs and used in Figure \ref{fg::illustration} in order to show that short CNVs are indeed common in real data.

\begin{table} [p] 
\caption{Offspring's CNVs detected from one/both of parents - For each of CNVs, starting SNP indices are presented along with the size in parentheses.} \label{tb::cnvs}
I. with father
\begin{center}
\begin{tabular} {cc rrr rrr rrr rr} \hline
Index &&\multicolumn{2}{c}{CBS} && \multicolumn{2}{c}{WBS} && \multicolumn{2}{c}{LRS} && \multicolumn{2}{c}{BWD}\\
 \cline{1-1} \cline{3-4} \cline{6-7} \cline{9-10}  \cline{12-13}
1&&        &  	   && 22369   & (7) && 22369	&   (7) &&  22369	&   (7) \\
2&&        &       &&         &     && 76529	&   (3) &&  76529	&   (3) \\
3&&163327  & (18)  &&         &     &&163331	&  (14) && 163331	&  (15) \\
4&&252509  &  (2)  && 252509  & (2) &&252509    &   (2) &&          &       \\
5&&        &  	   &&         &     &&325625	&   (6) &&        	&       \\
6&& 359377 & (11)  && 359372  & (14)&&       	&       &&        	&       \\
7&&        &       &&         &     &&379916	&   (2)	&& 379916	&   (2) \\
8&&392433  &  (5)  &&         &     &&392433	&   (5)	&& 392433	&   (4) \\
9&&507037 &  (6)  && 507039  & (4) &&507038	&   (5)	&& 507038	& (122) \\
10&&       &       && &  &&561443	&   (2)	&& 561443	&   (2) \\ \hline
\end{tabular}
\end{center}

II. with mother
\begin{center}
\begin{tabular} {cc rrr rrr rrr rr} \hline
Index &&\multicolumn{2}{c}{CBS} && \multicolumn{2}{c}{WBS} && \multicolumn{2}{c}{LRS} && \multicolumn{2}{c}{BWD}\\
\cline{1-1} \cline{3-4} \cline{6-7} \cline{9-10} \cline{12-13}
1&&       &       && & &&130814	&  (10)	&& 130814	&   (7) \\
2&&       &       && & &&228119	&   (6)	&& 228119	&   (6) \\
3&&       &       && & &&277328	&   (2)	&& 277328	&   (2) \\
4&&363744 & (9)	  && 363744 & (9) &&363744	&   (9)	&& 363744	&   (9) \\
5&&       &		  && 414247 & (3) &&414247	&   (3)	&& 414247	&   (3) \\
6&&       &    	  && & &&457597	&   (3)	&& 457597	&   (3) \\
7&&       &       && 519555 & (10)&&519555	&  (10)	&& 519555	&  (10) \\
8&&532211 &  (7)  && 532210 & (8) &&532211	&   (7) && 532210	&   (8) \\ \hline
\end{tabular}
\end{center}

III. with both father and mother
\begin{center}
\begin{tabular} {cc rrr rr ccc }  \hline
Index && \multicolumn{2}{c}{LRS} && \multicolumn{2}{c}{BWD} && \multicolumn{2}{c}{Comments} \\ \cline{1-1} \cline{3-4} \cline{6-7} \cline{9-10}
1&& 53949	&   (2) &&  53949	&   (2) && &\\
2&&152827	&   (2) && 152827	&   (2) && &\\
3&&359377	&  (11) && 359377	&   (9) && \multicolumn{2}{c}{CBS missed father.}\\
4&&442247	&   (4) && 442243	&   (8) && \multicolumn{2}{c}{WBS missed mother.}\\
5&&547470	& (146) && 547459	& (157) && \multicolumn{2}{c}{WBS missed mother.}\\
\hline
\end{tabular}
\end{center}
\end{table} 

The genetic information is to be inherited from parents to offsprings and can be utilized for validation of the detected CNVs. Table \ref{tb::cnvs} lists all the offspring's CNVs that are also detected from one/both of the parents. All the CNVs in Table \ref{tb::cnvs} are nearly, if not exactly, identical to the corresponding ones detected from the parents and thus those CNVs are considered as truth. We would like to emphasize that most of the true CNVs are quite short and both of the CBS and WBS miss many of them while the LRS and BWD miss only 1 and 3, respectively. We claim that some jointly detected CNVs from (at least one of) parents and offspring are still suspicious as false, if only a very minor portion of the detected CNVs overlapped compared to their entire length. The LRS detects 9 such suspicious CNVs while the CBS, WBS, and BWD detect 1, 1, and 2, respectively.

In summary, from the real data analysis for the trio SNP array, the LRS tends to call too many CNVs that includes large number of false positive while the CBS and WBS miss some true shorts CNVs, and we can conclude that the proposed BWD outperforms all others. This is concordant to the findings in the simulation studies in Section \ref{s:sim}.

\subsection{Read Depth from NGS Sequencing Data}
We further illustrate the (original) BWD on the RD data from high-throughtput sequencing data on chromosome 19 of a HapMap Yoruban female sample (NA19240) from the 1000 Genomes Project. The RD $y_i$ of the $i$th locus where $i = 1, \cdots, 54,361,060$ is adjusted by the guanine-cytosine (GC) content. Although the data can be used to analyze genomic variants in higher resolution with raw measure, as aforementioned the observations are highly flexible due to complicated sequencing process and requires a proper normalization/transformation. In order to handle these difficulties, we consider a local-median transformation as motivated by \cite{tony2012robust}. In particular, we firstly partition the RD data into small bins with size $M$, and then apply BWD for the sequence of the medians of observations in each bins. The transformed data sequence is then well-approximated by a normal distribution regardless the underlying distribution of the original data. If $M$ is large, the data are more accurately approximated by normal model, but the CNV shorter than $M$ bps cannot be accurately identified. (i.e., $M$ is a minimal resolution). As shown in Section \ref{s:sim} the BWD is not overly-sensitive to violation of the normal assumption and we set a relatively small number of $M = 100$ in the analysis.

For the BWD we set $\alpha = 0.05$ and the cutoff value is computed under the normal assumption. The BWD called fifteen CNVs. Figure \ref{fg::rd} provides zoom-in plots of some of CNVs identified by the BWD. The proposed method works reasonably  well for the (NGS) read-depth data as well after simple transformation.

Many existing CNV analysis tools for  high-throughput NGS data employ the CBS as a primal tool for calling CNVs. See \cite{duan2013comparative} and reference therein.
We would like to remark that the BWD can be a desirable  alternative under the presence of short CNVs hardly detected by the CBS.

\begin{figure}[p]
\centering
\subfigure{\includegraphics[width = 0.45\textwidth]{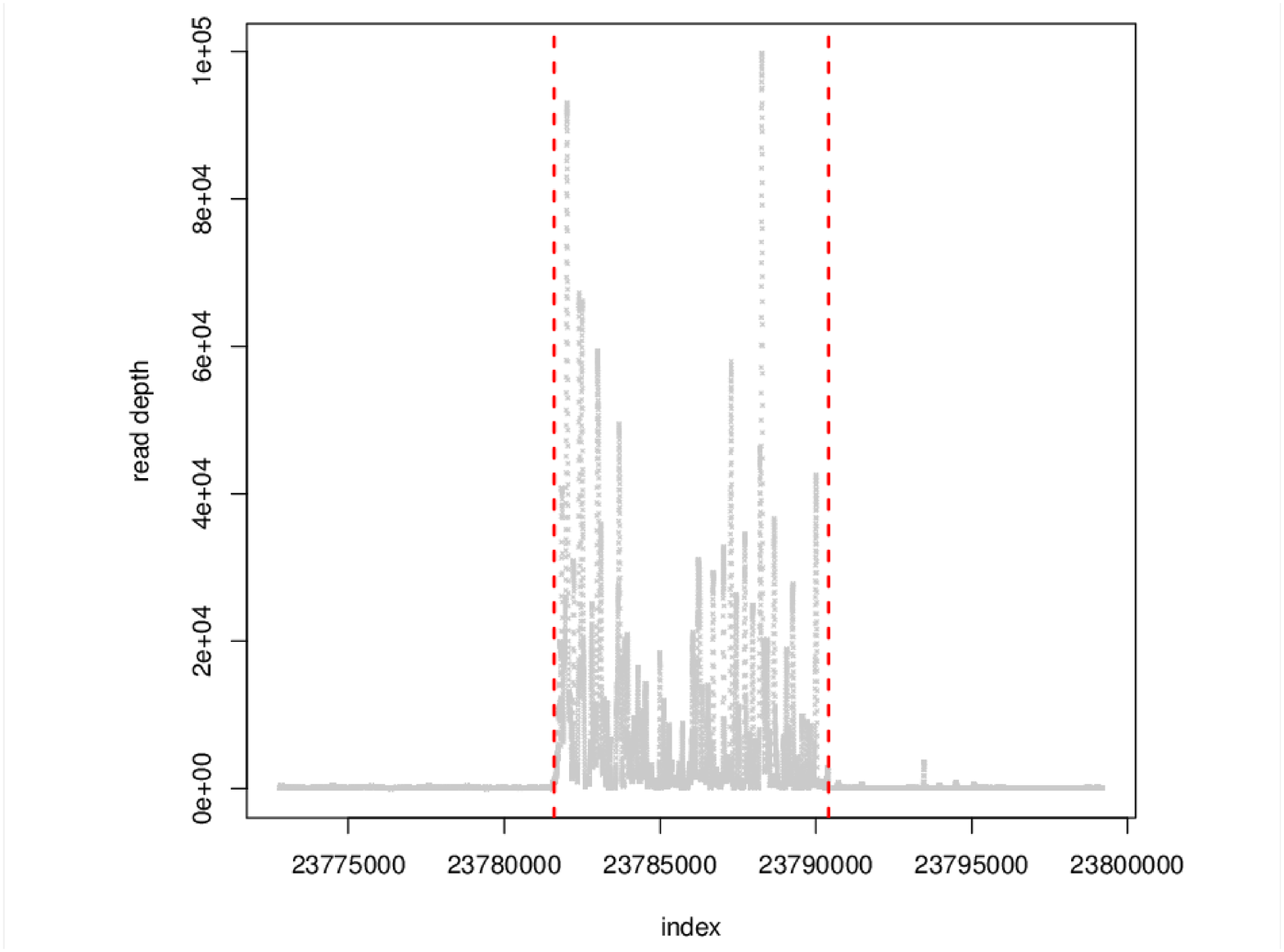}}
\subfigure{\includegraphics[width = 0.45\textwidth]{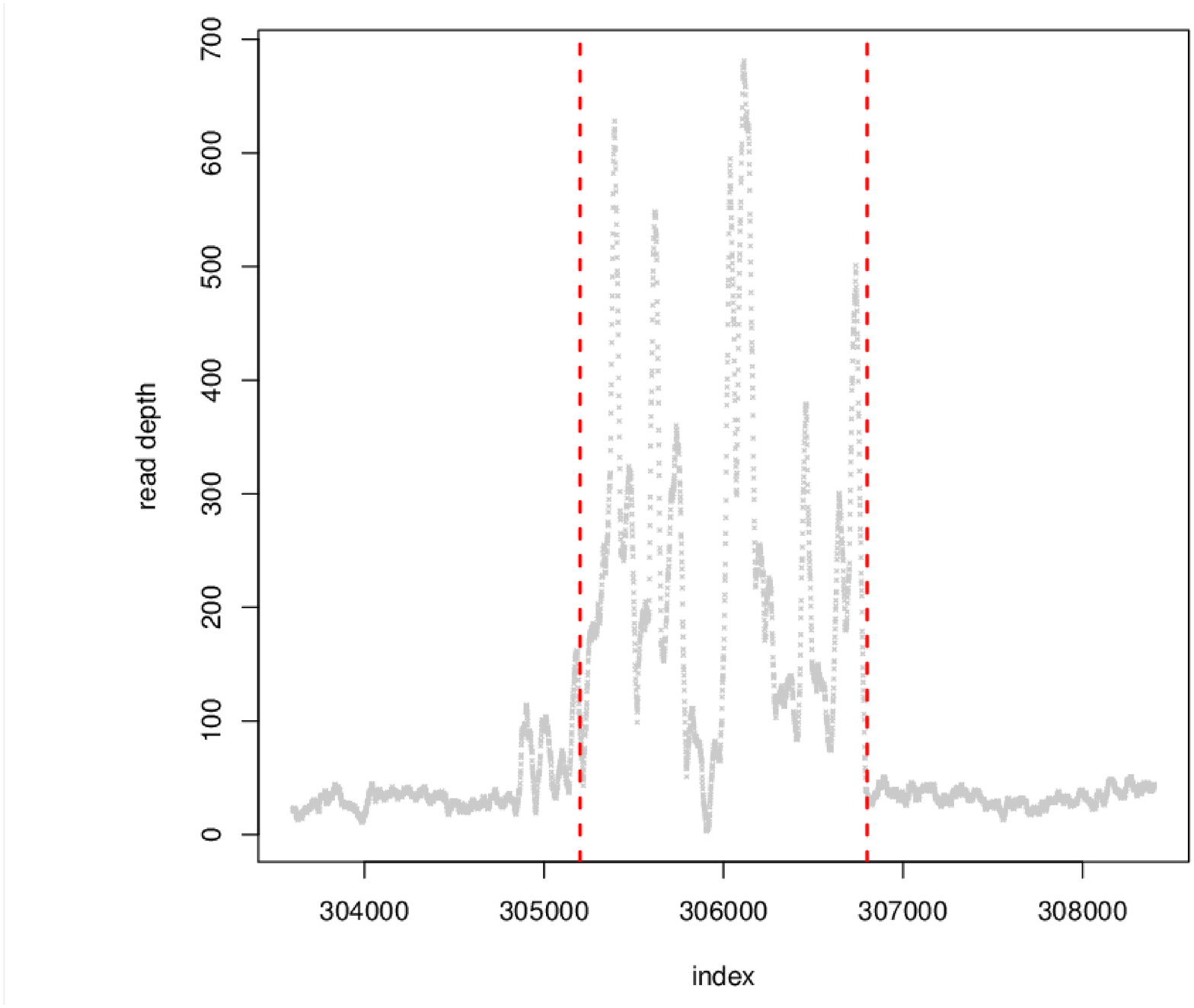}}
\subfigure{\includegraphics[width = 0.45\textwidth]{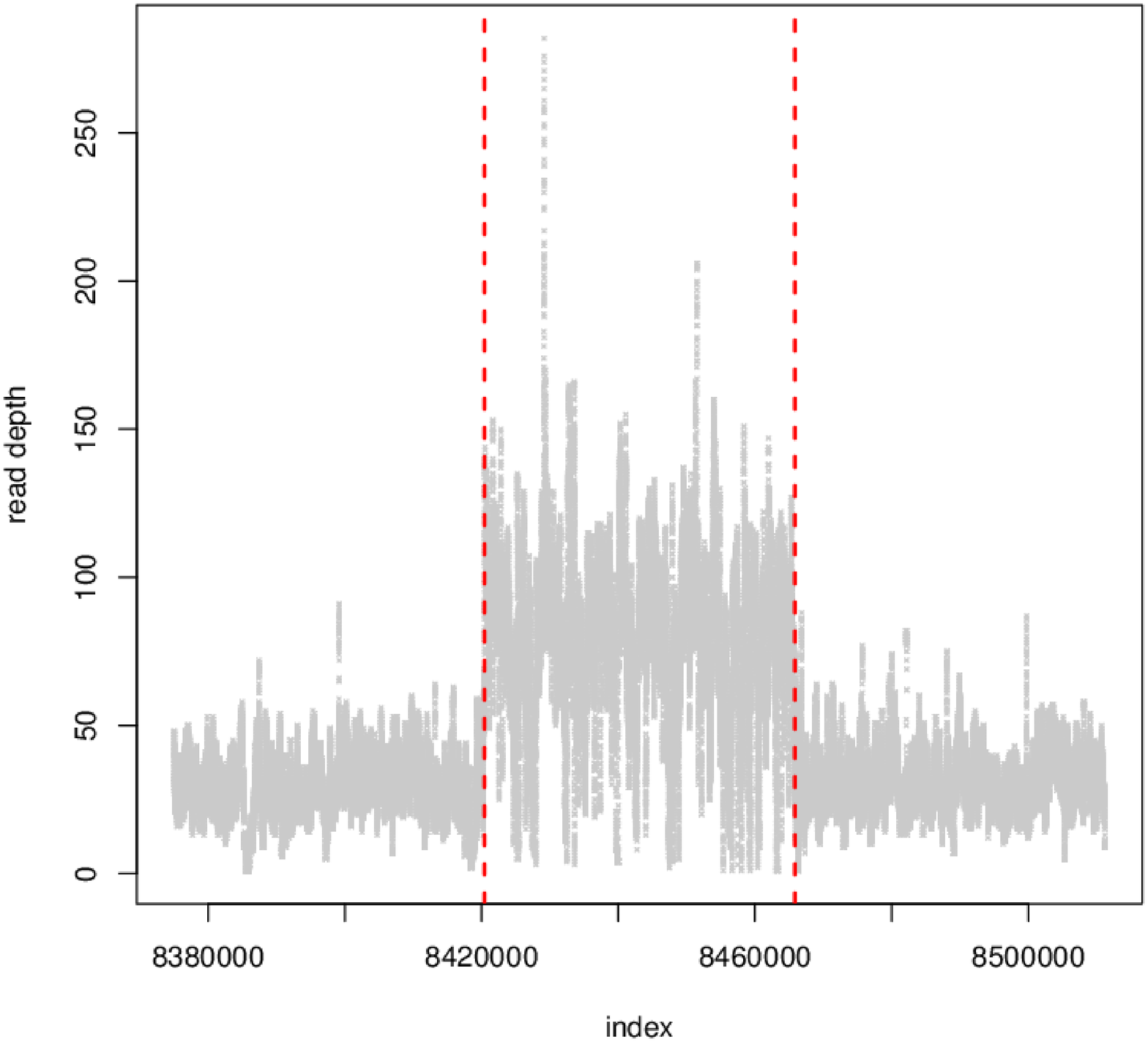}}
\subfigure{\includegraphics[width = 0.45\textwidth]{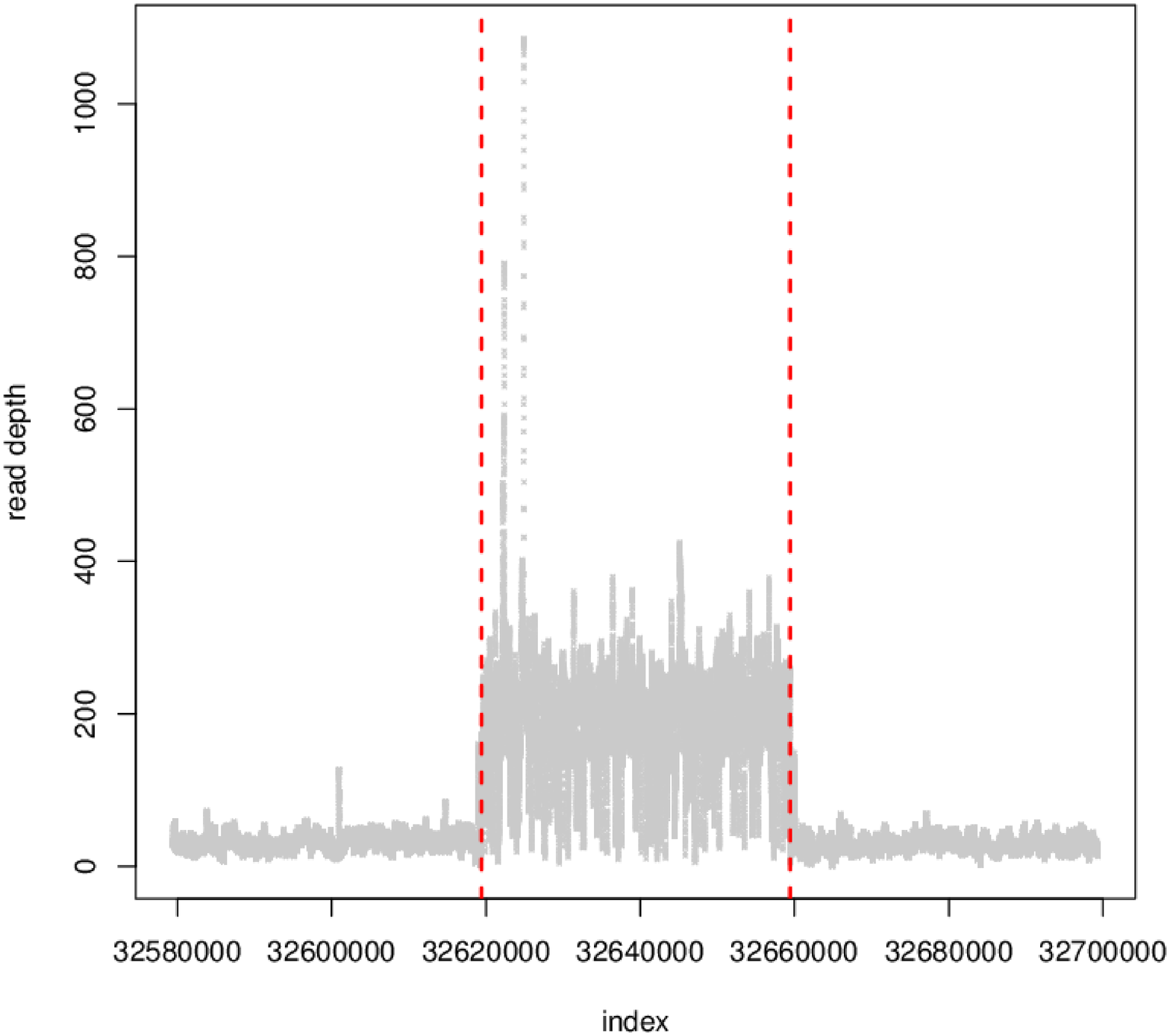}}
\subfigure{\includegraphics[width = 0.45\textwidth]{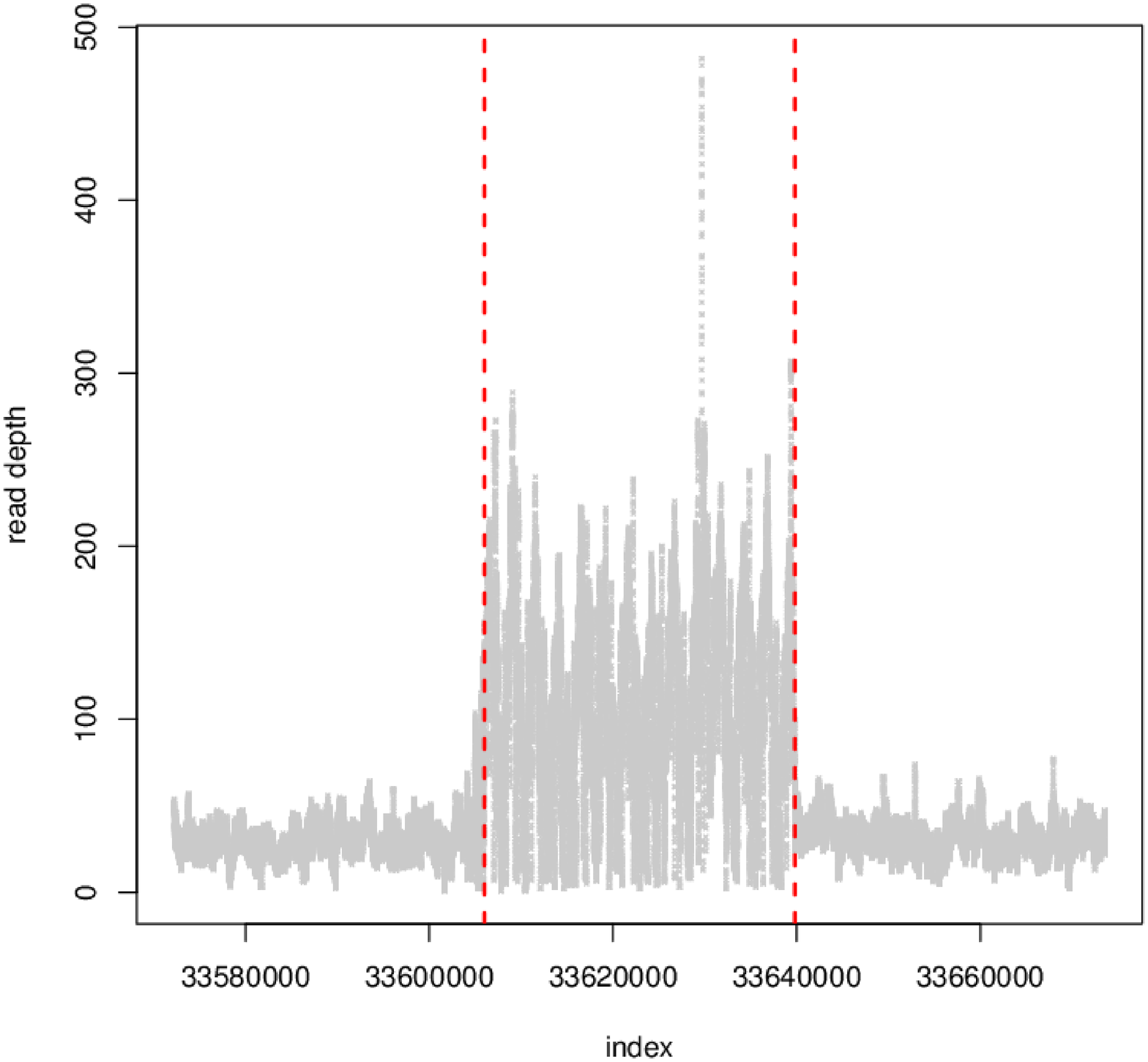}}
\subfigure{\includegraphics[width = 0.45\textwidth]{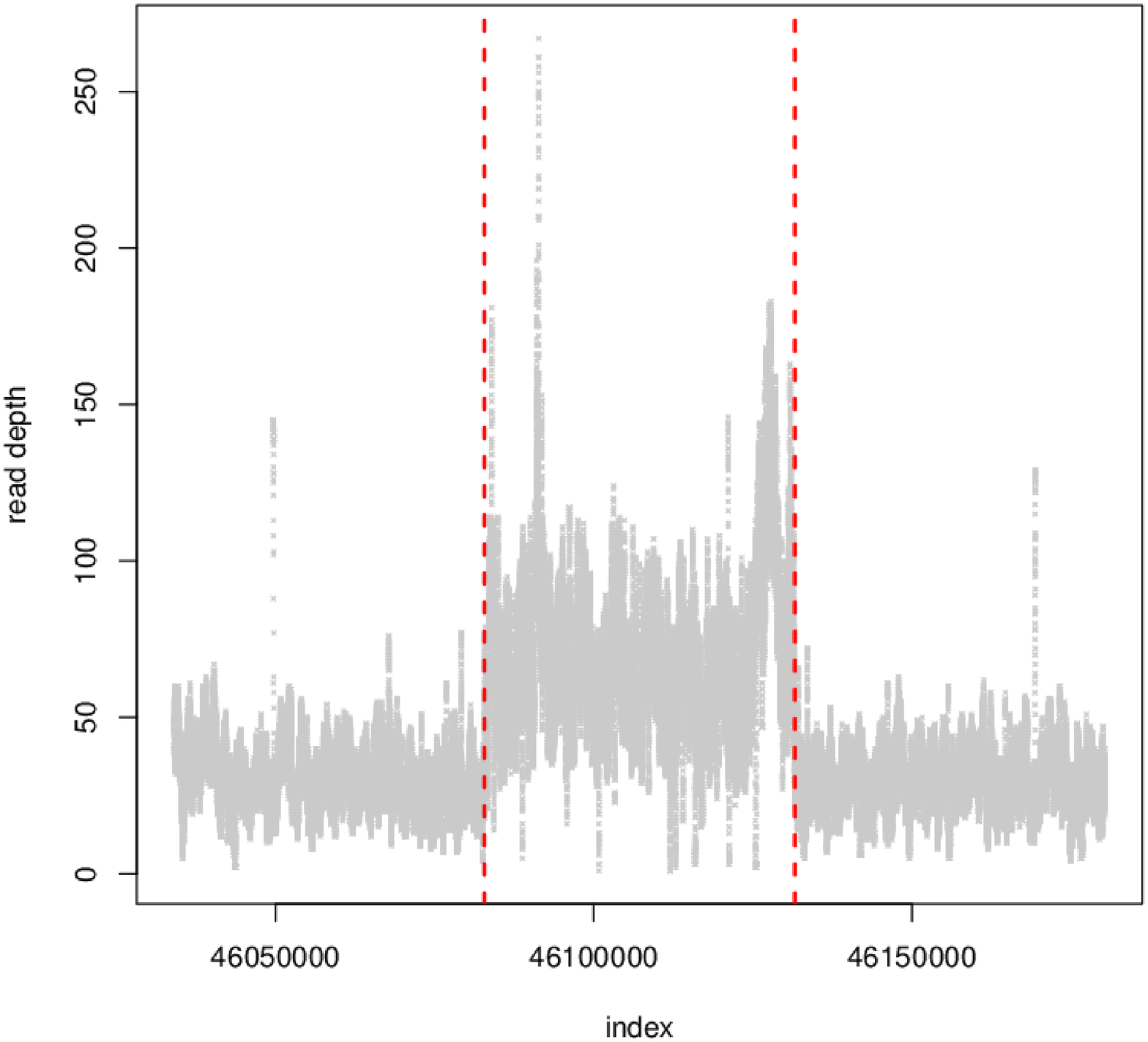}}
\caption{Zoom-in plots of the CNVs identified by the backward detection from the NGS read-depth data: The CNVs can be detected directly from the NGS read-depth after simple local median transformation.}\label{fg::rd}
\end{figure}

\section{Discussion} \label{s:discussion}
We propose a BWD procedure for change-point detection and apply it to CNV detection. The proposed BWD is a simple procedure that can be readily employed for high-dimensional data, but it still performs very well as illustrated to both simulated and real data especially when the true signals of interest are short, which is often the case in CNV detection problem. Similar to the CBS, the BWD is a general approach for change-point detection problems that can be used in various applications besides CNV detection from which it originally motivated, since it does not depend on any application-specific assumption.

The simple idea of the proposed BWD provides a possibility of further extension in various ways. Firstly, the gain of backward procedure compared to forward detection including CBS is obvious for short signal detection. However, the forward detection also has a clear benefit when the true signal is long and the mean-change is minor. Thus we can select either of the two depending on application. Moreover, we can develop a method that hybrids between the forward and backward detections analogous to the stepwise variable selection in the regression context. The idea is straightforward but require additional effort to improve computational efficiency especially for CNV applications. Next, we can extend BWD with different loss functions other than the squared $L_2$ loss. For example, the absolute deviance error can be used as a reasonable alternative under the presence of outliers. It is also possible to generalize the idea to more complex structure such as graph \citep{chen2015graph} by introducing a proper loss function defined on the space on the complex data object. Finally, as motivated by trio data, the backward idea can be extended to detect common signals shared by multiple sequences of observations.

\bibliographystyle{dcu}
\bibliography{references}
\end{document}